\def \VersionSpringer {}
\ifdefined\LaTeXdiff
	\let\VersionWithComments\undefined
	\let\WithReply\undefined
\fi
\documentclass[sttt]{svjour}
\smartqed  

\makeatletter
\def\cl@chapter{\@elt {theorem}}
\makeatother

\usepackage[utf8]{inputenc}
\usepackage[english]{babel}

\usepackage{amsfonts}
\usepackage{amsmath} %
\usepackage{amssymb} %

\usepackage{amsthm}
\theoremstyle{plain}
\usepackage[ruled,vlined,linesnumbered]{algorithm2e}

\usepackage{subcaption}

\usepackage{paralist} %

\ifdefined\VersionWithComments
	\usepackage{draftwatermark}
	\SetWatermarkText{draft}
	\SetWatermarkScale{15}
	\SetWatermarkColor[gray]{0.9}
\fi
\usepackage[svgnames,table]{xcolor}
\definecolor{darkblue}{rgb}{0.0,0.0,0.6}
\definecolor{darkgreen}{rgb}{0, 0.5, 0}
\definecolor{darkpurple}{rgb}{0.7, 0, 0.7}
\definecolor{violetcurie}{RGB}{115,26,67}
\definecolor{forestgreen}{rgb}{0.13,0.54,0.13}
\definecolor{darkblue}{rgb}{0, 0, 0.7}

\usepackage[
		colorlinks=true,
	\ifdefined \VersionWithComments
		pagebackref=true,
	\fi
		citecolor=darkgreen,
		linkcolor=darkblue,
		urlcolor=darkpurple,
	]{hyperref}

\usepackage[capitalise,english,nameinlink]{cleveref} %
\crefname{line}{\text{line}}{\text{lines}} %

\usepackage{pdflscape} %

\usepackage{tikz}
\usetikzlibrary{arrows,automata}
\tikzstyle{every node}=[initial text=]
\tikzstyle{location}=[rectangle, rounded corners, minimum size=12pt, draw=black, fill=blue!10, inner sep=2pt]
\tikzstyle{invariant}=[draw=black, dotted, inner sep=1pt] %

\definecolor{loccolor1}{rgb}{1, 0.3, 0.3}
\definecolor{loccolor2}{rgb}{0.3, 1, 0.3}
\definecolor{loccolor3}{rgb}{0.3, 0.3, 1}
\definecolor{loccolor4}{rgb}{1, 0.3, 1}
\definecolor{loccolor5}{rgb}{1, 1, 0.3}
\definecolor{loccolor6}{rgb}{0.3, 1, 1}
\definecolor{loccolor7}{rgb}{0.9, 0.6, 0.2}
\definecolor{loccolor8}{rgb}{0.7, 0.4, 1}
\definecolor{loccolor9}{rgb}{0.5, 1, 0.75}
\definecolor{loccolor10}{rgb}{0.8, 0.7, 0.6}
\definecolor{loccolor11}{rgb}{0.6, 0.7, 0.8}
\definecolor{loccolor12}{rgb}{0.2, 0.5, 0.9}
\definecolor{loccolor13}{rgb}{0.5, 0.9, 0.2}
\definecolor{loccolor14}{rgb}{0.9, 0.2, 0.5}
\definecolor{loccolor15}{rgb}{0.7, 0.7, 0.7}
\definecolor{loccolor16}{rgb}{0.8, 0.8, 0.5}

\newcommand{\styleact}[1]{\ensuremath{\textcolor{coloract}{\mathrm{#1}}}}
\newcommand{\styleclock}[1]{\ensuremath{\textcolor{colorclock}{#1}}}

\newcommand{\styleloc}[1]{\ensuremath{\mathrm{#1}}}
\newcommand{\styleparam}[1]{\ensuremath{\textcolor{colorparam}{#1}}}

\ifdefined\VersionSpringer
	\definecolor{coloract}{rgb}{0., 0., 0.}
	\definecolor{colorclock}{rgb}{0., 0., 0.}
	\definecolor{colordisc}{rgb}{0., 0., 0.}
	\definecolor{colorloc}{rgb}{0., 0., 0.}
	\definecolor{colorparam}{rgb}{0., 0., 0.}

	\newcommand{\cellHeader}[1]{\cellcolor{blue!20}\textbf{#1}}

	\newcommand{\couleurDec}{green!50}
	\newcommand{\couleurUndec}{red!50}
	\newcommand{\couleurOpen}{yellow!50}
	\newcommand{\couleurOpenRef}{black!5}
	\newcommand{\couleurN}{pink!50}
	\newcommand{\couleurNbounded}{pink!50}
	\newcommand{\couleurQ}{purple!50}
	\newcommand{\couleurQR}{purple!50}
	\newcommand{\couleurR}{purple!50}

	\newcommand{\couleurSyntaxInt}{blue!25}
	
	\newcommand{\couleurSyntaxIneqLarge}{blue!35}
	\newcommand{\couleurSyntaxJLR}{blue!20}
\else
	\definecolor{coloract}{rgb}{0.50, 0.70, 0.30}
	\definecolor{colorclock}{rgb}{0.4, 0.4, 1}
	\definecolor{colordisc}{rgb}{1, 0, 1}
	\definecolor{colorloc}{rgb}{0.4, 0.4, 0.65}
	\definecolor{colorparam}{rgb}{1, 0.6, 0.0}

	\newcommand{\cellHeader}[1]{\cellcolor{blue!40}\textbf{#1}}

	\newcommand{\couleurDec}{green}
	\newcommand{\couleurUndec}{red}
	\newcommand{\couleurOpen}{yellow}
	\newcommand{\couleurOpenRef}{black!10}
	\newcommand{\couleurN}{pink}
	\newcommand{\couleurNbounded}{pink}
	\newcommand{\couleurQ}{purple}
	\newcommand{\couleurQR}{purple}
	\newcommand{\couleurR}{purple}

	\newcommand{\couleurSyntaxInt}{blue!50}
	
	\newcommand{\couleurSyntaxIneqLarge}{blue!70}
	\newcommand{\couleurSyntaxJLR}{blue!40}
\fi

\newcommand{\cellDec}{\cellcolor{\couleurDec}}
\newcommand{\cellUndec}{\cellcolor{\couleurUndec}}
\newcommand{\cellOpen}{\cellcolor{\couleurOpen}open}
\newcommand{\cellOpenref}{\cellcolor{\couleurOpenRef}}
\newcommand{\cellN}{\cellcolor{\couleurN}\grandn{}}
\newcommand{\cellNbounded}{\cellcolor{\couleurNbounded}$\grandn$ bounded}
\newcommand{\cellQ}{\cellcolor{\couleurQ}\grandqplus{}}
\newcommand{\cellQboundedonetwo}{\cellcolor{\couleurQ}$\grandqplus_{[1;2]}$}
\newcommand{\cellQR}{\cellcolor{\couleurQR}\grandqplus{}/\grandrplus{}}
\newcommand{\cellR}{\cellcolor{\couleurR}\grandrplus{}}

\newcommand{\cellSyntaxEq}{\cellcolor{\couleurSyntaxInt}$\clock = \param|\constant$}
\newcommand{\cellSyntaxComp}{\cellcolor{\couleurSyntaxInt}$\clock \compOp \param|\constant$}
\newcommand{\cellSyntaxOpen}{\cellcolor{\couleurSyntaxInt}$\clock \opstrict \param$}

\newcommand{\cellSyntaxIneqLarge}{\cellcolor{\couleurSyntaxIneqLarge}$\clock \leq\geq \param|\constantNonneg$}
\newcommand{\cellSyntaxJLRg}{\cellcolor{\couleurSyntaxJLR}\ensuremath{\clock \compOp \plterm}} %
\newcommand{\cellSyntaxNone}{\cellcolor{gray!50}None}

\ifdefined \VersionWithComments
	\usepackage{marginnote}
	\newcommand{\marginX}{\marginnote{\huge{\quad\quad\textbf{!}\quad\quad}}}
	\newcommand{\ea}[1]{\mbox{}{\color{violet}\marginX{}\textbf{[\'Etienne}: #1]}}
	\newcommand{\instructions}[1]{{\color{blue}\marginX{}\textbf{[\'Instructions: ``#1'']}}}
	\newcommand{\reviewer}[2]{\mbox{}{\color{red}\marginX{}\textbf{[Reviewer #1}: ``#2'']}}
	\newcommand{\todo}[1]{\mbox{}{\color{red}{\marginX{}\textbf{TODO}\ifx#1\\\else:\ \fi #1}}} %
\else
	\newcommand{\instructions}[1]{}
	\newcommand{\ea}[1]{}
	\newcommand{\reviewer}[2]{}
	\newcommand{\todo}[1]{}
\fi

\def \VersionLong {}

\ifdefined \VersionLong
	\newcommand{\LongVersion}[1]{#1}
	\newcommand{\ShortVersion}[1]{}
\else
	\newcommand{\LongVersion}[1]{\ifdefined\VersionWithComments{\color{black!40}#1}\fi}
	\newcommand{\ShortVersion}[1]{\ifdefined\VersionWithComments{\color{red!40!black}#1}\else#1\fi}
\fi

\newcommand{\A}{{\ensuremath{\mathsf A}}}

\newcommand{\init}{_0}

\newcommand{\Actions}{\Sigma}
\newcommand{\action}{a}
\newcommand{\C}{C}

\newcommand{\Clock}{X} %
\newcommand{\ClockCard}{H} %
\newcommand{\clock}{x} %
\newcommand{\ClocksZero}{\vec{0}}
\newcommand{\clockval}{w} %
\newcommand{\compOp}{\bowtie}
\newcommand{\compOpLeq}{\trianglelefteq} %
\newcommand{\constant}{d}
\newcommand{\constantNonneg}{d^+}
\newcommand{\constantTimelaps}{\delta}
\newcommand{\counterValue}{c}
\newcommand{\edge}{e}
\newcommand{\Edges}{E}
\newcommand{\longuefleche}[1]{\stackrel{#1}{\mapsto}}
\newcommand{\longueflecheRel}{{\mapsto}}
\newcommand{\fleche}[1]{\stackrel{#1}{\rightarrow}}
\newcommand{\flecheRel}{{\rightarrow}}

\newcommand{\guard}{g}

\newcommand{\invariant}{I}
\newcommand{\K}{K}
\newcommand{\loc}{l} %
\newcommand{\locinit}{\loc\init}
\newcommand{\Loc}{L} %
\newcommand{\LocFinal}{F} %
\newcommand{\Param}{P} %
\newcommand{\param}{p} %
\newcommand{\ParamCard}{M} %
\newcommand{\plterm}{\mathit{plt}}
\newcommand{\Problem}{\ensuremath{\mathcal{P}}}
\newcommand{\pval}{v} %
\newcommand{\Q}{Q}

\newcommand{\resets}{R}
\newcommand{\sinit}{s\init} %
\newcommand{\somelocs}{G} %
\newcommand{\state}{s} %
\newcommand{\States}{S} %
\newcommand{\timeBound}{T}

\newcommand{\varproblem}{\ensuremath{\phi}}

\newcommand{\ceil}[1]{\lceil #1 \rceil}
\newcommand{\floor}[1]{\lfloor #1 \rfloor}

\newcommand{\reset}[2]{\ensuremath{[#1]_{#2}}}
\newcommand{\valuate}[2]{\ensuremath{#2(#1)}}

\newcommand{\DomainP}{{\ensuremath{\mathbb P}}}
\newcommand{\DomainT}{{\ensuremath{\mathbb T}}}
\newcommand{\grandn}{{\ensuremath{\mathbb N}}}
\newcommand{\grandq}{{\ensuremath{\mathbb Q}}}
\newcommand{\grandqplus}{{\ensuremath{\grandq^+}}} %
\newcommand{\grandr}{{\ensuremath{\mathbb R}}}
\newcommand{\grandrplus}{{\ensuremath{\grandr^+}}} %
\newcommand{\grandz}{{\mathbb Z}}
\newcommand{\mitlzeroinf}{{\ensuremath{\mathsf{MITL}_{0,\infty}}}}
\newcommand{\pmitlzeroinf}{{\ensuremath{\mathsf{PMITL}_{0,\infty}}}}

\newcommand{\BuchiEF}{Büchi}

\newcommand{\opstrict}{\mathrel{<>}}

\ifdefined\LaTeXdiff
 
\fi

\newcommand{\hytech}{\textsc{HyTech}}
\newcommand{\imitator}{\textsf{IMITATOR}}
\newcommand{\pat}{PAT}
\newcommand{\PSyHCoS}{PSyHCoS}
\newcommand{\romeo}{\textsc{Roméo}}
\newcommand{\uppaal}{\textsc{Uppaal}}

\newcommand{\defProblem}[3]
{%
\noindent\fcolorbox{black}{blue!15}{
	\begin{minipage}{.95\columnwidth}
		\textbf{#1 problem:}\\
		\textsc{Input}: #2\\
		\textsc{Problem}: #3
	\end{minipage}	
}
	\smallskip

}

\ifdefined \VersionWithComments
 	\definecolor{colorok}{RGB}{80,80,150}
\else
	\definecolor{colorok}{RGB}{0,0,0}
\fi

\newcommand{\eg}{\textcolor{colorok}{e.\,g.,}\xspace}
\newcommand{\ie}{\textcolor{colorok}{i.\,e.,}\xspace}
\newcommand{\viz}{\textcolor{colorok}{viz.,}\xspace}
\newcommand{\wrt}{\textcolor{colorok}{w.r.t.}\xspace}

	\usepackage[backend=biber,backref=true,style=alphabetic,url=false,doi=true,defernumbers=true,sorting=anyt,maxnames=99]{biblatex} %
	\renewbibmacro*{doi+eprint+url}{%
		\iftoggle{bbx:doi}
			{\color{black!40}\footnotesize\printfield{doi}}
			{}%
		\newunit\newblock
		\iftoggle{bbx:eprint}
			{\usebibmacro{eprint}}
			{}%
		\newunit\newblock
		\iftoggle{bbx:url}
			{\usebibmacro{url+urldate}}
			{}%
	}

\begin{document}

\title{What's decidable about parametric timed automata?}
\author{\'Etienne Andr\'e%
\thanks{%
	This is the author version of the manuscript of the same name published in the International Journal on Software Tools for Technology Transfer (\href{https://link.springer.com/journal/10009}{STTT}), April 2019, Volume 21, Issue 2, pp 203--219.
	The final version is available at \href{https://www.doi.org/10.1007/s10009-017-0467-0}{\nolinkurl{10.1007/s10009-017-0467-0}}.
	This work is partially supported by the ANR national research program PACS (ANR-14-CE28-0002).}
}                     %
\institute{Université Paris 13, LIPN, CNRS, UMR 7030, F-93430, Villetaneuse, France}
\date{Received: 3 August 2016 / Revised version: 23rd July 2017}
\maketitle

	\thispagestyle{plain}

\ifdefined \VersionWithComments
	\textcolor{red}{\textbf{This is the version with comments. To disable comments, comment out line~3 in the \LaTeX{} source.}}
\fi

\begin{abstract}
	Parametric timed automata (PTAs) are a powerful formalism to reason, simulate and formally verify critical real-time systems.
	After 25 years of research on PTAs, it is now well-understood that any non-trivial problem studied is undecidable for general PTAs.
	We provide here a survey of decision and computation problems for PTAs.
	On the one hand, bounding time, bounding the number of parameters or the domain of the parameters does not (in general) lead to any decidability.
	On the other hand, restricting the number of clocks, the use of clocks (compared or not with the parameters), and the use of parameters (\eg{} used only as upper or lower bounds) leads to decidability of some problems.
	We also put emphasis on open problems.
	We also discuss formalisms close to parametric timed automata (such as parametric hybrid automata or parametric interrupt timed automata), and we study tools dedicated to PTAs and their extensions.
\end{abstract}

\keywords{decidability, decision problems, parametric timed model checking, parameter synthesis, L/U-PTAs, hybrid automata}

\section{Introduction}\label{section:introduction}

The absence of undesired behaviors in real-time critical systems is of utmost importance in order to ensure the system safety.
Model checking aims at formally verifying a model of the system against a correctness property.
Timed automata (TAs) are a popular formalism to model and verify safety critical systems with timing constraints.
TAs extend finite state automata with clocks, \ie{} real-valued variables increasing linearly~\cite{AD94}.
These clocks can be compared with integer constants in guards (sets of linear inequalities that must be satisfied to take a transition) and invariants (sets of linear inequalities that must be satisfied to remain in a location).
TAs have been widely studied (see \eg{} \cite{AM04}), and several state-of-the-art model checkers (such as \uppaal{}~\cite{LPY97} or \pat{}~\cite{SLDP09}) support TAs as an input language.

TAs benefit from many interesting decidable properties, such as the emptiness of the accepted language, the reachability of a control state, etc.
	Other problems are undecidable though, such as the universality of the accepted timed language; in addition, given a TA, building a TA recognizing the complement of the timed language of the first TA cannot be achieved in general.
	TAs were also studied in a robust version, \ie{} when all timing guards can be enlarged or shrinked by an infinitesimal constant factor, without changing the language, the reachability of a control state, etc. (see \cite{Markey11,BMS13} for surveys).

However, TAs also suffer from some limitations.
First, they cannot be used to specify and verify systems incompletely specified (\ie{} whose timing constants are not known yet), and hence cannot be used in early design phases.
Second, verifying a system for a \emph{set} of timing constants usually requires to enumerate all of them one by one if they are supposed to be integer-valued; in addition, TAs cannot be used anymore to verify a system for a set of timing constants that are to be taken in a rational-
or real-valued dense interval.
Third, robustness in TAs often assumes that all guards can be enlarged or shrinked by the same small variation; considering independent variations or considering both enlarging and shrinking was not addressed. %

Parametric timed automata (PTAs) overcome these limitations by allowing the use of parameters (\ie{} unknown constants) in guards and invariants~\cite{AHV93}.
This increased expressive power comes at the price of the undecidability of most interesting problems---at least in the general case.

In this paper, we consider decision problems for PTAs proposed in the past 25~years.
On the one hand, bounding time, bounding the number of parameters or the domain of the parameters does not (in general) lead to any decidability.
On the other hand, restricting the number of clocks, the use of clocks (compared or not with the parameters), and the use of parameters (\eg{} used only as upper or lower bounds) can lead to the decidability of some problems.
In addition, an extension to parameters of some variants of timed automata benefit from some decidability results, such as reset-PTAs and parametric interrupt timed automata.

\paragraph{Related surveys}
To the best of our knowledge, no survey was dedicated specifically to decision problems for PTAs.
Moreover, in addition to numerous results in the past 25~years proved in various settings with different syntax and assumptions, recent results in the field in the past three years %
justify the need for a clear picture of these updated (un)decidability results.
Furthermore, surveying decision problems for PTAs has important practical implications as, for undecidable decision problems, the associated synthesis problems cannot be solved exactly.

Related works include \cite{AM04} that studies decidability results of timed automata.
In~\cite{Markey11,BMS13}, various problems related to the robustness in TAs are studied.
Then, \cite{HKPV98} is not a survey, but exhibits decidable subclasses of hybrid automata, an extension of timed automata where variables can have (in general) arbitrary rates.
Then, \cite{AMPS12} acts both as a survey and as a contribution paper that studies hybrid automata with ``low dimensions'', \ie{} with few variables.
Our survey is also concerned (in \cref{section:bounding}) with decidability results for PTAs with few variables (\ie{} clocks and parameters).

\paragraph{About this manuscript}
This manuscript is a revised and extended version of~\cite{Andre15}.
New results unpublished at the time of~\cite{Andre15} were added.
Moreover, \cref{table:dec:EFemptiness} was improved, and its description and summary was significantly enhanced.
In addition, two new sections were added: formalisms beyond PTAs are studied in \cref{section:beyond} and tools and applications of PTAs are reviewed in \cref{section:tools}.

\paragraph{Outline}
In \cref{section:preliminaries}, we propose a unified syntax for PTAs, and we define the decision problems that we will consider throughout this manuscript.
In \cref{section:undecidable}, we recall general undecidability results for PTAs.
We then study in \cref{section:bounding} the decidability when restricting the syntax of PTAs (number of variables, syntax of the constraints, etc.).
We consider specifically in \cref{section:LU} the subclass of PTAs where parameters must be used either always as lower-bounds or always as upper-bounds, namely L/U-PTAs.
Formalisms beyond PTAs, including parametric versions of hybrid automata, interrupt timed automata and time Petri nets, are studied in \cref{section:beyond}.
Tools supporting PTAs and known applications of PTAs are reviewed in \cref{section:tools}.
We conclude by emphasizing open problems in \cref{section:open}.

\section[Parametric timed automata and problems]{Parametric timed automata and problems}\label{section:preliminaries}
\subsection{Clocks, parameters and constraints}

Let $\grandz$, $\grandn$, $\grandqplus$ and $\grandrplus$ denote the sets of
	(possibly negative) integer numbers,
	(non-negative) natural numbers,
	non-negative rational numbers,
	and
	non-negative real numbers,
	respectively.
In the following, $\DomainT$ denotes the domain of time,
and $\DomainP$ the domain of the parameters;
these domains will be instantiated with $\grandn$, $\grandqplus$ or $\grandrplus$ subsequently.
Throughout this survey, let $\constant$ denote an integer constant in~$\grandz$, and $\constantNonneg$ denote a non-negative constant in~$\grandn$.

Let us assume a set~$\Clock = \{ \clock_1, \dots, \clock_\ClockCard \} $ of \emph{clocks}, that are $\DomainT$-valued variables that evolve at the same rate.
Let us assume a set~$\Param = \{ \param_1, \dots, \param_\ParamCard \} $ of \emph{parameters}, \ie{} unknown constants.
A parameter {\em valuation} $\pval$ is a function
$\pval : \Param \rightarrow \DomainP$. %
Throughout this survey, symbols $\clock$, $\clock_i$ denote clocks whereas $\param$, $\param_i$ denote parameters.

A \emph{parametric linear term} is $\sum_{1 \leq i \leq \ParamCard} \alpha_i \param_i + \constant$, with $\alpha_i \in \grandz$;
in the following $\plterm$ will denote a parametric linear term.

A \emph{(linear) inequality} is $\clock \compOp \plterm$, where $\clock$ is a clock, $\plterm$ a parametric linear term,
and \mbox{${\mathrel{\compOp}} \in \{ <, \leq, \geq, > \}$}.
We give in \cref{table:operators} the conventions used throughout this survey concerning comparison operators.
A \emph{(linear) constraint} is a conjunction of linear inequalities.
A \emph{(linear) diagonal constraint} is a conjunction of either linear inequalities, or \emph{linear diagonal inequalities} of the form $\clock_i - \clock_j \compOp \plterm$.

A \emph{simple inequality} is either $\clock \compOp \param$ or $\clock \compOp \constantNonneg$.
A \emph{simple constraint} is a conjunction of simple inequalities.

\begin{table}%
	\centering

	\begin{tabular}{| c | c |}
		\hline
		\cellHeader{Operator} & \cellHeader{Definition}\\
		\hline
		$\mathrel{\compOp} $ & $ \{ <, \leq, \geq, > \}$\\
		\hline
		$\mathrel{\leq\geq} $ & $ \{ \leq, \geq\}$\\
		\hline
		$\opstrict $ & $ \{ <, >\}$\\
		\hline
		$\mathrel{\compOpLeq} $ & $ \{ < , \leq \}$\\
		\hline
	\end{tabular}
	
	\caption{Syntax of operators in guards}
	\label{table:operators}
\end{table}

\subsection{A unified syntax for parametric timed automata}

The syntax of PTAs varies a lot in the literature; we give below a definition that should include most definitions in the literature, and at least all definitions of the papers considered in this survey.
That is, any definition of PTAs can be obtained from the following one by adding restrictions such as removing the set of accepting locations, forbidding invariants, restricting the domain of clocks or parameters, simplifying the syntax of the guards and invariants (\eg{} forbidding diagonal constraints), etc.

\begin{definition}[Parametric timed automaton]\label{def:PTA}
	A \emph{parametric timed automaton (PTA)} is a tuple \mbox{$\A = (\Sigma, \Loc, \locinit, \LocFinal, \Clock, \Param, \invariant, \Edges)$}, where:
	\begin{itemize}
		\item $\Sigma$ is a finite set of actions,
		\item $\Loc$ is a finite set of locations,
		\item $\locinit \in \Loc$ is the initial location,
		\item $\LocFinal \subseteq \Loc$ is a set of accepting (or final) locations,
		\item $\Clock$ is a set of clocks with domain $\DomainT = \grandrplus$,
		\item $\Param$ is a set of parameters with domain $\DomainP = \grandrplus$,
		\item $\invariant$ is the invariant, assigning to every $\loc\in \Loc$ a diagonal constraint $\invariant(\loc)$, and
		\item $\Edges$ is a set of edges  $(\loc,\guard,\action,\resets,\loc')$
		where
		$\loc,\loc'\in \Loc$ are the source and destination locations, $\guard$ is a diagonal constraint which is the transition guard, $\action \in \Sigma$, and $\resets\subseteq \Clock$ is a set of clocks to be reset.
	\end{itemize}
\end{definition}

Given a PTA~$\A$ and a parameter valuation~$\pval$, the \emph{valuation} of~$\A$ with~$\pval$, denoted by $\valuate{\A}{\pval}$, is the nonparametric PTA where each occurrence of~$\param$ is replaced with $\pval(\param)$.
If $\pval$ assigns an integer (or rational) value to each parameter, then $\valuate{\A}{\pval}$ is a~TA.
However, if some parameters are assigned to an irrational value, then $\valuate{\A}{\pval}$ belongs to the class of TAs with irrational constants, for which the reachability of a given location is undecidable~\cite{Miller00}.

A clock is said to be a \emph{parametric clock} if it is compared with at least one parameter in at least one guard or invariant; otherwise, it is a \emph{nonparametric clock}.
This notion is central when studying the decidability of problems for PTAs with few clocks and parameters.

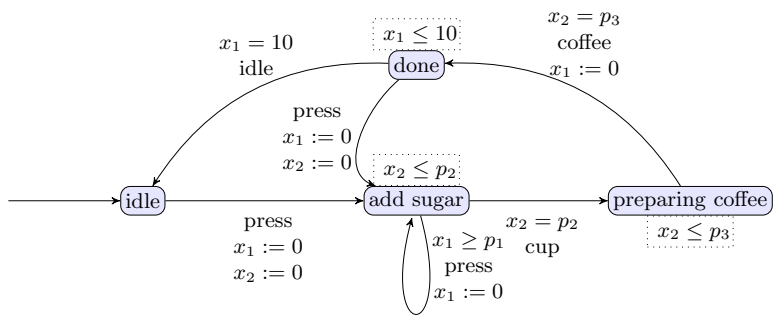
\begin{figure*}%
	\newcommand{\ratio}{0.5\textwidth}
 
	\centering

	\scalebox{.9}{
	\begin{tikzpicture}[scale=4, auto, ->, >=stealth']
 
		\node[location, initial] at (0,0) (idle) {\styleloc{idle}};
 
		\node[location] at (1,0) (add_sugar) {\styleloc{add\ sugar}};
		\node [invariant,above] at (add_sugar.north) {\begin{tabular}{@{} c @{\ } c@{} }& $ \styleclock{x_2} \leq \styleparam{\param_2}$\\\end{tabular}};
 
		\node[location] at (2,0) (preparing_coffee) {\styleloc{preparing\ coffee}};
		\node [invariant,below] at (preparing_coffee.south) {\begin{tabular}{@{} c @{\ } c@{} }& $ \styleclock{x_2} \leq \styleparam{\param_3}$\\\end{tabular}};
 
		\node[location] at (1,.5) (done) {\styleloc{done}};
		\node [invariant,above] at (done.north) {\begin{tabular}{@{} c @{\ } c@{} }& $ \styleclock{x_1} \leq 10$\\\end{tabular}};

		\path (idle) edge node[below]{\begin{tabular}{@{} c @{\ } c@{} }
		 & $\styleact{press}$\\
		 & $\styleclock{x_1}:=0$\\
		 & $\styleclock{x_2}:=0$\\
		\end{tabular}} (add_sugar);

		\path (add_sugar) edge[loop below] node[above right]{\begin{tabular}{@{} c @{\ } c@{} }
		& $ \styleclock{x_1} \geq \styleparam{\param_1}$\\
		 & $\styleact{press}$\\
		 & $\styleclock{x_1}:=0$\\
		\end{tabular}} (add_sugar);

		\path (add_sugar) edge node[below]{\begin{tabular}{@{} c @{\ } c@{} }
		& $ \styleclock{x_2} = \styleparam{\param_2}$\\
		 & $\styleact{cup}$\\
		\end{tabular}} (preparing_coffee);

		\path (preparing_coffee) edge[bend right] node[above]{\begin{tabular}{@{} c @{\ } c@{} }
		& $ \styleclock{x_2} = \styleparam{\param_3}$\\
		 & $\styleact{coffee}$\\
		 & $\styleclock{x_1}:=0$\\
		\end{tabular}} (done);

		\path (done) edge[out=220,in=160] node[left]{\begin{tabular}{@{} c @{\ } c@{} }
		 & $\styleact{press}$\\
		 & $\styleclock{x_1}:=0$\\
		 & $\styleclock{x_2}:=0$\\
		\end{tabular}} (add_sugar);
 
		\path (done) edge[bend right] node[above]{\begin{tabular}{@{} c @{\ } c@{} }
		& $ \styleclock{x_1} = 10$\\
		 & $\styleact{idle}$\\
		\end{tabular}} (idle);
	\end{tikzpicture}
	}
	\caption{A coffee machine modeled using a PTA}
	\label{fig:coffee}
\end{figure*}

\begin{example}
	Consider the coffee machine in \cref{fig:coffee}, modeled using a PTA with 4 locations, 2 clocks ($\styleclock{x_1}$ and~$\styleclock{x_2}$) and 3 parameters ($\styleparam{\param_1}, \styleparam{\param_2}, \styleparam{\param_3}$).
	Both clocks $\styleclock{x_1}$ and $\styleclock{x_2}$ are parametric clocks.
	No diagonal constraints are used (in fact all constraints are simple constraints).
	The machine can initially idle for an arbitrarily long time.
	Then, whenever the user presses the (unique) button (action \styleact{press}), the PTA enters location ``add sugar'', resetting both clocks.
	The machine can remain in this location as long as the invariant ($\styleclock{x_2} \leq \styleparam{\param_2}$) is satisfied;
	there, the user can add a dose of sugar by pressing the button (action \styleact{press}), provided the guard ($\styleclock{x_1} \geq \styleparam{\param_1}$) is satisfied, which resets~$\styleclock{x_1}$.
	That is, the user cannot press twice the button (and hence add two doses of sugar) in a time less than~$\styleparam{\param_1}$.
	Then, $\styleparam{\param_2}$ time units after the machine left the idle mode, a cup is delivered (action \styleact{cup}), and the coffee is being prepared;
	eventually, $\styleparam{\param_2}$ time units after the machine left the idle mode, the coffee (action \styleact{coffee}) is delivered.
	Then, after 10 time units, the machine returns to the idle mode---unless a user again requests a coffee by pressing the button.
\end{example}

\paragraph{Semantics}
\newcommand{\wv}[2]{#1|#2} %

The semantics of a PTA~$\A$ can be defined as the union over all parameter valuations~$\pval$ of the semantics of~$\valuate{\A}{\pval}$.
In the following, given $\constantTimelaps \in \grandrplus$, $\clockval + \constantTimelaps$ denotes the valuation such that $(\clockval + \constantTimelaps)(\clock) = \clockval(\clock) + \constantTimelaps$, for all $\clock \in \Clock$.
Given $\resets \subseteq \Clock$, we define the \emph{reset} of a clock valuation~$\clockval$, denoted by $\reset{\clockval}{\resets}$, as the valuation resetting the clocks in~$\resets$, and keeping the other clocks unchanged.
Given a rational parameter valuation~$\pval$, $\valuate{\C}{\pval}$ denotes the constraint over~$\Clock$ obtained by replacing each parameter~$\param$ in~$\C$ with~$\pval(\param)$.
Likewise, given a clock valuation~$\clockval$, $\valuate{\valuate{\C}{\pval}}{\clockval}$ denotes the expression obtained by replacing each clock~$\clock$ in~$\valuate{\C}{\pval}$ with~$\clockval(\clock)$.
We use the notation $\wv{\clockval}{\pval} \models \C$ to indicate that $\valuate{\valuate{\C}{\pval}}{\clockval}$ evaluates to true.
We write $\ClocksZero$ for the clock valuation that assigns~$0$ to all clocks.

\begin{definition}[Concrete semantics of a TA]
	Given a PTA $\A = (\Sigma, \Loc, \locinit, \LocFinal, \Clock, \Param, \invariant, \Edges)$,
	and a rational parameter valuation~\(\pval\),
	the concrete semantics of $\valuate{\A}{\pval}$ is given by the timed transition system $(\States, \sinit, \flecheRel)$, with
	\begin{itemize}
		\item $\States = \{ (\loc, \clockval) \in \Loc \times \grandrplus^\Clock \mid \wv{\clockval}{\pval} \models \invariant(\loc) \}$,%
		\item $\sinit = (\locinit, \ClocksZero) $,
		\item $\flecheRel$ consists of the discrete and (continuous) delay transition relations:
				\begin{itemize}
			\item discrete transitions: $(\loc,\clockval) \fleche{\edge} (\loc',\clockval')$, %
				if $(\loc, \clockval) , (\loc',\clockval') \in \States$, there exists $\edge = (\loc,\guard,\action,\resets,\loc') \in \Edges$,
					$\clockval' = \reset{\clockval}{\resets}$
					and $\wv{\clockval}{\pval} \models \guard$.
			\item delay transitions: $(\loc,\clockval) \fleche{\constantTimelaps} (\loc, \clockval+\constantTimelaps)$, with $\constantTimelaps \in \grandrplus$, if $\forall \constantTimelaps' \in [0, \constantTimelaps], (\loc, \clockval+\constantTimelaps') \in \States$.
		\end{itemize}
	\end{itemize}
\end{definition}

Moreover we write $(\loc, \clockval)\longuefleche{\edge} (\loc',\clockval')$ for a sequence of delay and discrete transitions where
	$((\loc, \clockval), \edge, (\loc', \clockval')) \in \longueflecheRel$ if
		$\exists \constantTimelaps, \clockval'' :  (\loc,\clockval) \fleche{\constantTimelaps} (\loc,\clockval'') \fleche{\edge} (\loc',\clockval')$.

Given a TA~$\valuate{\A}{\pval}$ with concrete semantics $(\States, \sinit, \flecheRel)$,
we refer to the states of~$\States$ as the \emph{concrete states} of~$\valuate{\A}{\pval}$.
A concrete run of~$\valuate{\A}{\pval}$ is an alternating sequence of concrete states of $\valuate{\A}{\pval}$ and edges starting from the initial concrete state $\sinit$ of the form 
$\sinit \longuefleche{\edge_0} \state_1\longuefleche {\edge_1} \cdots \longuefleche{\edge_{m-1}} \state_m$, such that for all $i = 0, \dots, m-1$, $\edge_i \in \Edges$, $\state_{i+1} \in \States$, and $(\state_i , \edge_i , \state_{i+1}) \in \longueflecheRel$.
Given a concrete state~$\state=(\loc, \clockval)$, we say that $\state$ is reachable (or that $\valuate{\A}{\pval}$ reaches $\state$) if $\state$ belongs to a concrete run of $\valuate{\A}{\pval}$.
By extension, we say that $\loc$ is reachable in~$\valuate{\A}{\pval}$.
A run is \emph{maximal} if it is either infinite, or cannot be extended by any discrete transition (possibly after some delay transition).

A finite run $\sinit \longuefleche{\edge_0} \state_1\longuefleche {\edge_1} \cdots \longuefleche{\edge_{m-1}} (\loc_m, \clockval_m)$ is \emph{accepting} if $\loc_m \in \LocFinal$.

The accepted timed language is the set of timed words (alternating sequences of actions and time elapsing) associated with an accepting run, \ie{} a run ending in a location of $\LocFinal$ (or, in some works, passing infinitely often by a location in~$\LocFinal$).
Note that some works make a difference between finite and infinite runs.
The untimed language of a TA is the timed language projected onto the actions.
The set of traces (or trace set) is the set of accepting runs projected onto the locations and actions, \ie{} a set of alternating locations and actions.
This is a nonstandard definition of traces (compared to \eg{} \cite{vanGlabbeek90}), but we keep this term as it is used in, \eg{} \cite{ACEF09,AM15}.

A \emph{symbolic semantics} is also defined for PTAs in~\cite{HRSV02,ACEF09,JLR15} as a parametric zone graph, where a symbolic state is made of a discrete part (the current location) and a symbolic, continuous part (a set of diagonal constraints, \ie{} $\clock_i - \clock_j \compOp \plterm$, sometimes allowing disjunctions).
\todo{donner ou sans intérêt ?}

\paragraph{Simple PTAs}
We define \emph{simple PTAs} as the subclass of PTAs where guards and invariants are simple constraints.
We propose this class to show that, even in this restricted situation, all non-trivial problems are undecidable (\cref{section:undecidable}).

\paragraph{Variants of the PTA syntax}
PTAs were first defined in~\cite{AHV93} using a set of accepting locations.
This is similar to timed automata~\cite{AD94}.
\emph{Timed safety automata} %
were introduced later in~\cite{HNSY94} by removing the final locations, but adding invariants to locations; many subsequent papers then refer to timed safety automata as simply ``timed automata''.
When timed automata with accepting locations are equipped with Büchi conditions (to be accepting, an infinite timed word must pass infinitely often through at least one of the accepting locations), they are referred to as timed Büchi automata. %
It was shown that the timed expressive power of timed safety automata is strictly less than that of timed Büchi automata~\cite{HKW95}.

The syntax of %
PTAs differs in most of the papers in the literature.
Concerning guards and invariants,
	in~\cite{AHV93} (resp.~\cite{Miller00}), guards (resp.\ guards and invariants) are conjunctions of inequalities of the form $\clock \compOp \param$.
	In~\cite{HRSV02,BlT09}, guards are conjunctions of inequalities of the form $\clock_i - \clock_j \compOpLeq \plterm \cup \{\infty \}$;
		in~\cite{HRSV02} invariants have the same form as guards (invariants are not considered in~\cite{BlT09}).
		In~\cite{ACEF09}, any linear constraint over $\Clock \cup \Param$ is allowed in guards and invariants.
	In~\cite{Doyen07}, guards and invariants are all open, \ie{} of the form $\clock \opstrict \param$ or $\clock \opstrict \constantNonneg$.
	In~\cite{JLR15,ALR16FORMATS}, guards and invariants are conjunctions of inequalities of the form $\clock \compOp \plterm$; in addition, in~\cite{JLR15} invariants can only bound clocks from above (\ie{} $\clock \compOpLeq \plterm$).
	In~\cite{BBLS15}, guards are conjunctions of inequalities of the form $\clock \compOp \param$ and invariants can only bound clocks from above (\ie{} $\clock \compOpLeq \param$).
	In~\cite{AM15}, guards and invariants are conjunctions of inequalities of the form $\clock \compOp \param + \constant$, $\clock \compOp \constantNonneg$ or $\param \compOp \constant$ (although the proofs of undecidability only need inequalities of the form $\clock \compOp \param $ or $\clock \compOp \constantNonneg$).
	In~\cite{ALR15,ALR16ICFEM}, guards and invariants are conjunction of simple inequalities.

A set of accepting locations is considered in~\cite{AHV93,BlT09,BBLS15,ALR16FORMATS}, but only~\cite{BlT09} is interested in infinite accepting runs, \ie{} runs that pass infinitely often by an accepting location; hence this latter work considers what could be referred to as parametric timed Büchi automata.
In contrast, \cite{Miller00,HRSV02,Doyen07,ACEF09,JLR15,AM15,ALR16ICFEM} consider parametric timed safety automata (\ie{} without accepting locations).

\begin{remark}\label{remark:invariants}
	The restriction that the invariants can only bound clocks from above (\ie{} $\clock \compOpLeq \plterm$) is not a real restriction:
	in timed automata, invariant that bound clocks from below (\ie{} $\constant \compOpLeq \clock $) can be moved to all incoming edges.
	The same applies to PTAs.
	In other words, papers defining PTAs requiring invariants to use only invariants with clocks bounded from above are equivalent to PTAs with no restrictions at all on the invariants.
\end{remark}

\paragraph{Expressiveness}
A comparison of the expressiveness of these different syntactic models remains to be done.
Whereas it is likely that allowing constraints of the form $\clock \compOp \plterm$ may be simulated using constraints of the form $\clock \compOp \param \mid \constantNonneg$ (perhaps adding additional locations, clocks and parameters), the expressiveness may differ when adding a set of accepting locations.
In fact, the expressiveness of a PTA was not even defined, until we recently proposed two first possible definitions~\cite{ALR16FORMATS}:
	the expressiveness of a PTA~$\A$ (with accepting locations) is either the union over all parameter valuations of the accepted untimed words (``untimed language of~$\A$''), or the union over all parameter valuations of pairs made of an accepted untimed word and the associated valuation (``constrained untimed language of~$\A$'').
Then, several subclasses of PTAs are compared \wrt{} these two definitions.

However, no comparison of the syntax used in guards and invariants was proposed.
A challenging future work would be to show that a PTA with constraints of the form $\clock \compOp \plterm$ can be for example translated into an equivalent PTA with constraints of the form $\clock \compOp \param \mid \constantNonneg$ at the cost of~$n$ additional clocks and/or parameters.

\subsection{Decision and computation problems}\label{ss:problems}

Following the presentation in~\cite{JLR15},
given a class of decision problems \Problem{} (reachability, unavoidability, etc.), let us define the \Problem{}-emptiness, the \Problem{}-universality and the \Problem{}-finiteness.
\smallskip

\defProblem
	{\Problem-emptiness}
	{A PTA~\A{} and an instance $\varproblem$ of \Problem{}}
	{Is the set of parameter valuations $\pval$ such that $\valuate{\A}{\pval}$ satisfies $\varproblem$ empty?}

\defProblem
	{\Problem-universality}
	{A PTA~\A{} and an instance $\varproblem$ of \Problem{}}
	{Are all parameter valuations $\pval$ such that $\valuate{\A}{\pval}$ satisfies $\varproblem$?}

\defProblem
	{\Problem-finiteness}
	{A PTA~\A{} and an instance $\varproblem$ of \Problem{}}
	{Is the set of parameter valuations $\pval$ such that $\valuate{\A}{\pval}$ satisfies $\varproblem$ finite?}

\smallskip

In this survey, we mainly focus on reachability and unavoidability properties, and call them EF and AF respectively.\footnote{%
	The names EF, AF, EG, AG were first used for PTAs in~\cite{JLR15}, and come from the CTL syntax.
}
For example, given a PTA~$\A$ and a subset~$\somelocs$ of its locations\footnote{%
	In general, it can be handful to set $\somelocs = \LocFinal$; but as not all definitions of PTAs in the literature have accepting locations, we use here the set~$\somelocs$ to denote goal locations.
}, EF-emptiness asks:
	``is the set of parameter valuations~$\pval$ such that at least one location of~$\somelocs$ is reachable in $\valuate{\A}{\pval}$ empty?''
And AF-universality asks: ``are all parameter valuations~$\pval$ such that any location in~$\somelocs$ is unavoidable in $\valuate{\A}{\pval}$?''
We will also mention the EG property, that checks whether there exists a maximal run along which the locations remain in\ShortVersion{~a subset}~$\somelocs$\ShortVersion{~of the locations},
and the AG property that checks whether the locations remain in~$\somelocs$ for all runs.\footnote{%
	Note that EF-, AF-, EG-, and AG-emptiness are equivalent to
	AG-, EG-, AF-, EF-universality,
	respectively.
}

Additionally, we will consider the language (resp.\ trace) preservation (emptiness) problem~\cite{AM15}:
given a PTA~$\A$ and a parameter valuation~$\pval$, does there exist another valuation~$\pval' \neq \pval$ such that the untimed languages (resp.\ sets of traces) of $\valuate{\A}{\pval}$ and $\valuate{\A}{\pval'}$ are the same?

\ShortVersion{%
We finally define the \Problem-synthesis problem:
	Given a PTA~\A{} and an instance $\varproblem$ of \Problem{},
	compute the parameter valuations such that $\valuate{\A}{\pval}$ satisfies $\varproblem$.
}
We finally define the following computation problem:

\smallskip

\defProblem
	{\Problem-synthesis}
	{A PTA~\A{} and an instance $\varproblem$ of \Problem{}}
	{Compute the parameter valuations such that $\valuate{\A}{\pval}$ satisfies $\varproblem$.}

\smallskip
For example, given a PTA~$\A$ and a subset~$\somelocs$ of its locations, EF-synthesis consists in synthesizing parameter valuations $\pval$ such that at least one location of~$\somelocs$ is reachable in $\valuate{\A}{\pval}$ from the initial state.

\begin{example}
	Let us exemplify some decision and computation problems for the PTA in \cref{fig:coffee}.
	Assume the unique target location is ``done'', \ie{} $\somelocs=\{ \styleloc{done} \}$.
	EF-emptiness asks whether the set of parameter valuations that can reach location ``\styleloc{done}'' for some run is empty; this is false (\eg{} $\styleparam{\param_1} = 1$, $\styleparam{\param_2} = 2$, $\styleparam{\param_3} = 3$ can reach ``\styleloc{done}'').
	EF-universality asks whether all parameter valuations can reach location ``\styleloc{done}'' for some run; this is false (no parameter valuation such that $\styleparam{\param_2} > \styleparam{\param_3}$ can reach ``done'').
	AF-emptiness asks whether the set of parameter valuations that can reach location ``\styleloc{done}'' for all runs is empty; this is false (\eg{} $\styleparam{\param_1} = 1$, $\styleparam{\param_2} = 2$, $\styleparam{\param_3} = 3$ cannot avoid ``\styleloc{done}'').
	EF-synthesis consists in synthesizing all valuations for which a run reaches location ``\styleloc{done}''; the resulting set of valuations is $0 \leq \styleparam{\param_2} \leq \styleparam{\param_3} \leq 10 \land \styleparam{\param_1} \geq 0$.\ea{en fait, seulement param positive ici}
\end{example}

\section[Almost everything is undecidable for simple PTAs]{Almost everything is undecidable for simple PTAs}\label{section:undecidable}

In this entire section, we consider simple PTAs without restriction on the number of clocks and parameters.
In that situation, all non-trivial problems studied in the literature are undecidable, with the exception of the membership problem (that asks whether the language of a valuated PTA is empty)---which is rather a problem for TAs.
By non-trivial, we mean requiring a semantic analysis, and not, \eg{} a sole analysis of the syntax of the PTA (\eg{} ``is the number of clocks even'', or any problem defined in \cref{ss:problems} by setting $\somelocs = \Loc$).

We also show that bounding time (\cref{ss:b-time}) or bounding the parameter domain for rational-valued parameters (\cref{ss:b-domain}) preserves the undecidability.
However, we will show in \cref{section:bounding} that bounding the number of clocks and/or parameters brings decidability.

All proofs of undecidability reduce from either the halting problem, or the boundedness problem, of a 2-counter machine, both known to be undecidable~\cite{Minsky67}.

\subsection{Decidability of the membership}
In~\cite{AHV93}, the membership problem for PTAs is defined as follows:
given a PTA~$\A$ and a parameter valuation~$\pval$, is the language of $\valuate{\A}{\pval}$ empty?
The membership problem is not strictly speaking a problem for PTAs, but rather for TAs, since it considers a valuated PTA.
As a consequence, the decidability of this problem only relies on known results for~TAs.

On the one hand,
the membership problem is decidable (and PSPACE-complete) for PTAs
	over discrete time ($\DomainT = \grandn$ and $\DomainP = \grandn$),
	over dense time with integer-valued parameters ($\DomainT = \grandrplus$ and $\DomainP = \grandn$),
	and over dense time with rational-valued parameters ($\DomainT = \grandrplus$ and $\DomainP = \grandq$)~\cite{AD94}.

On the other hand, the membership problem becomes undecidable with real-valued (in fact irrational) parameters.\ea{OUI, c'est la bonne réf : le ``paramètre'' irrationnel est une constante FIXÉE (quoique arbitraire)}
Indeed, the reachability of a location in a TA with irrational constants is undecidable~\cite{Miller00}.
The idea is to encode a 2-counter machine using 2 clocks~$\clock_1$ and~$\clock_2$ (plus an additional third clock), where the value~$\counterValue_i$ of counter~$i$ is encoded using $\clock_i = \counterValue_i \times \tau$, for $i \in \{ 1, 2 \}$, with $\tau$ the irrational constant (the value $\sqrt{2}$ is suggested for~$\tau$).

\subsection{General undecidable problems}
\paragraph{EF-emptiness}
The seminal paper on PTAs~\cite{AHV93} showed that the EF-emptiness problem is undecidable for PTAs, both over discrete time, and over dense time (real-valued clocks and real-valued parameters).
The proof consists in reducing from the halting problem of a 2-counter machine.
The idea of the encoding of the 2-counter machine is to use parameters (the value of which can be arbitrarily large) to encode the maximum value of the counters.
Although not explicitly stated in~\cite{AHV93}, the proof of undecidability also works for real-valued clocks with integer-valued parameters.

\paragraph{AF-emptiness}
In~\cite{JLR15}, it is proved that the AF-emptiness is undecidable for L/U-PTAs (a subclass of PTAs, see \cref{section:LU}) with 3 clocks and 4 integer-valued parameters, and hence for PTAs as well.
Again, the proof of undecidability consists in reducing from the halting problem of a 2-counter machine.
Another proof is provided in~\cite{ALR16ICFEM} that uses 3 clocks and only 2 rational-valued parameters.

\paragraph{AG-emptiness}
In~\cite{ALR16ICFEM}, it is proved that the AG-emptiness problem is undecidable with 3 clocks and 2 rational-valued parameters.

\paragraph{EG-emptiness}
In~\cite{ALime17}, it is proved that the EG-emptiness problem is undecidable with 4 clocks and 3 parameters.

\begin{remark}
	For all three previous problems (AF-emptiness, AG-emptiness and EG-emptiness), the result is in fact proved for a subclass of PTAs---namely L/U-PTAs for AF-emptiness and EG-emptiness, and bounded integer-points PTAs (see \cref{ss:IPPTA}) for AG-emptiness---so the number of clocks and parameters needed for the encoding is certainly not minimal for general PTAs, and might therefore be reduced using smarter constructions.
\end{remark}

\begin{remark}
	Note that the undecidability of all of these problems rules out the possibility to perform exact parametric model checking of CTL-like properties on PTAs.
\end{remark}

\paragraph{Language and trace preservation problems}

Both the language preservation and the trace preservation problems are undecidable for simple PTAs~\cite{AM15}.
The \emph{continuous} (or robust) versions of those problems additionally require that the language (resp.\ set of traces) is preserved under any intermediary valuation of the form $\lambda \cdot \pval + (1-\lambda)\cdot \pval'$, for $\lambda\in[0,1]$ (with the classical definition of addition and scalar multiplication).

The language preservation problems and its continuous version are undecidable for a PTA with at least 4~parametric clocks~\cite{AM15}.

The trace preservation and its continuous version are undecidable too; the proof of this result comes with three flavors:
\begin{enumerate}
	\item the first proof involves diagonal constraints (\ie{} of the form $\clock_i - \clock_j \compOp \plterm$, which goes beyond the syntax of simple PTAs), but only a fixed number of parametric clocks~\cite{AM15};
	\item the second proof does not involve diagonal constraints.
	It involves a bounded number of locations (but with an unbounded number of transitions) and an unbounded number of parametric clocks; by unbounded we mean not constant but depending on the size of the counter-machine~\cite{AM15};
	\item the third proof uses a bounded number of clocks and parameters, and an unbounded number of locations~\cite{ALM18}.
\end{enumerate}
The need for an unbounded number of clocks in the first two versions of this proof comes from the fact that the proof encodes the 2-counter machine with a fixed number of locations (to reduce easily from language preservation to trace preservation), which thus requires to encode each location with a different clock.
Note that the first two versions of the proof are, to the best of our knowledge, the only attempt to model a 2-counter machine using PTAs with a constant number of locations (at the cost of an unbounded number of clocks).

\subsection{Bounding time}\label{ss:b-time}

Bounded-time model checking consists in checking a property \emph{within a bounded time domain}.
That is, we assume a predefined time bound (say $\timeBound$), and we only consider the system behavior in the time interval $[0, \timeBound]$.
Undecidable problems might become decidable in this situation, or be of a lower complexity.
For example, the language inclusion for timed automata becomes decidable over bounded-time~\cite{OW10}, although it is undecidable in general.
In addition, time-bounded reachability becomes decidable for a special subclass of hybrid automata with monotonic (either non-negative or non-positive) rates~\cite{BDGORW13}, although it is undecidable in general.

In contrast, the EF-emptiness problem remains undecidable for (general) PTAs over bounded, dense time~\cite[Theorem 3.4]{Jovanovic13}.

This said, we emphasize that (quite trivially) model checking \emph{discrete-time} PTAs over bounded-time would become decidable; the same is likely to hold for \emph{dense-time} PTAs with \emph{integer-valued} parameters over bounded-time.
(This remains to be shown formally though.)
\subsection{Bounding the parameter domain}\label{ss:b-domain}

Bounding the parameter domain consists in setting a minimal and a maximal (non-infinite) bound on the possible parameter valuations of a PTA.

\paragraph{Decidability for integer-valued parameters}
For integer parameters, any problem for a PTA over a bounded parameter domain is decidable iff the corresponding problem is decidable for a TA.
In fact, the $\Problem$-emptiness problem for PTAs with bounded integer is PSPACE-complete for any class of problems $\Problem$ that is PSPACE-complete for TAs~\cite{JLR15}.
Indeed, it suffices to enumerate all parameter valuations, of which there is a finite number.
As a consequence, EF-, AF-, EG-, AG-emptiness are all decidable; and so are language and trace preservation.
More generally, the whole TCTL model checking, including reachability and unavoidability, is PSPACE-complete~\cite{ACD93}, and therefore the corresponding emptiness problems are PSPACE-complete for PTAs with bounded integer parameters.

In~\cite{JLR15}, a symbolic method is proposed to compute EF- and AF-synthesis; experiments showed that this symbolic computation is faster than an exhaustive enumeration (using \uppaal{}).

\paragraph{Undecidability for rational-valued parameters}
For rational-valued parameters, the EF-emptiness problem is undecidable for a single parameter in $[1,2]$~\cite{Miller00}.
EG-emptiness~\cite{ALime17}, AF- and AG-emptiness~\cite{ALR16ICFEM}, as well as language and trace preservation~\cite{AM15} are also undecidable for one or two rational-valued bounded parameter(s) (typically bounded by~$[0,1]$).

\section[Bounding the numbers of clocks and parameters]{Bounding the numbers of clocks and parameters}\label{section:bounding}
\subsection{EF-emptiness}

Since the seminal paper on PTAs~\cite{AHV93}, the decidability of the EF-emptiness problem was studied in various settings, by bounding the number of parametric clocks, of nonparametric clocks, and of parameters.
The syntax was also restrained.

\begin{table*}%

	\centering
	\footnotesize
	\setlength{\tabcolsep}{1pt}
\scalebox{1}{
	\begin{tabular}{| c | c | c | c | c | c | c | c | c |}
		\hline
		\cellHeader{\DomainT{}} & \cellHeader{}\DomainP{} & \cellHeader{Guards} & \cellHeader{Invariants} & \cellHeader{P-clocks} & \cellHeader{NP-clocks} & \cellHeader{Params} & \cellHeader{Decidability} & \cellHeader{Main ref.}\\
		\hline
		\cellN{} & \cellN{} & \multicolumn{2}{c|}{\cellSyntaxComp{}} &1 & 0 & fixed & \cellDec{}(at most) PTIME & \cite{Miller00} (consequence)\\
		\hline
		\cellN{} & \cellN{} & \multicolumn{2}{c|}{\cellSyntaxComp{}} &1 & 0 & any & \cellDec{}(at most) NP-complete & \cite{Miller00} (consequence)\\

		\hline
		\cellN{} & \cellN{} & %
			\multicolumn{2}{c|}{\cellSyntaxComp{}}
			& 1 & any & any & \cellDec{}NEXPTIME-complete & \cite{BO14,BBLS15} (consequence)\\

		\hline
		\cellN{} & \cellN{} & \multicolumn{2}{c|}{\cellSyntaxIneqLarge{}} & 2 & any & 1 & \cellDec{}PSPACE$^\mathrm{NEXP}$-hard & \cite{BO14}\\
		\hline

		\cellN{} & \cellN{} & \multicolumn{2}{c|}{any} & 2 & any & $>1$ & \cellOpen{} & \cellOpenref{}\\
		\hline

		\cellN{} & \cellN{} & \cellSyntaxComp{} & \cellSyntaxNone{} & 3 & 0 & 1 & \cellUndec{}undecidable & \cite{BBLS15}\\

		\hline
		\cellN{} & \cellN{} & \cellSyntaxEq{} & \cellSyntaxNone{} & 3 & 0 & 6 & \cellUndec{}undecidable & \cite{AHV93}\\

		\hline
		\cellN{} & \cellN{} & \multicolumn{2}{c|}{\cellSyntaxOpen{}} & any & any & any & \cellOpen{} & \cellOpenref{}\\
		\hline
		\cellN{} & \cellNbounded{} & %
			\multicolumn{2}{c|}{\cellSyntaxJLRg{}}
			&any & any & any & \cellDec{}(at most) PSPACE-complete & \cite{JLR15} (consequence)\\
		\hline
		\cellR{} & \cellN{} & \multicolumn{2}{c|}{\cellSyntaxComp{}} &1 & 0 & fixed & \cellDec{}(at most) PTIME & \cite{Miller00} (consequence)\\
		\hline
		\cellR{} & \cellN{} & \multicolumn{2}{c|}{\cellSyntaxComp{}} &1 & 0 & any & \cellDec{}(at most) NP-complete & \cite{Miller00} (consequence)\\

		\hline
		\cellR{} & \cellN{} & %
			\multicolumn{2}{c|}{\cellSyntaxComp{}}
			& 1 & any & any & \cellDec{}NEXPTIME & \cite{BBLS15}\\

		\hline
		\cellR{} & \cellN{} & \multicolumn{2}{c|}{any} & 2 & any & any & \cellOpen{} & \cellOpenref{}\\

		\hline
		\cellR{} & \cellN{} & \cellSyntaxComp{} & \cellSyntaxNone{} & 3 & 0 & 1 & \cellUndec{}undecidable & \cite{BBLS15}\\
		\hline
		\cellR{} & \cellN{} & \cellSyntaxEq{} & \cellSyntaxNone{} & 3 & 0 & 6 & \cellUndec{}undecidable & \cite{AHV93} (consequence)\\

		\hline
		\cellR{} & \cellN{} & \multicolumn{2}{c|}{\cellSyntaxOpen{}} & any & any & any & \cellOpen{} & \cellOpenref{}\\
		\hline
		\cellR{} & \cellNbounded{} & %
			\multicolumn{2}{c|}{\cellSyntaxJLRg{}}
			& any & any & any & \cellDec{}PSPACE-complete & \cite{JLR15}\\
		\hline
		\cellR{} & \cellQ{} & \multicolumn{2}{c|}{\cellSyntaxComp{}} &1 & 0 & fixed & \cellDec{}PTIME & \cite{Miller00}\\
		\hline
		\cellR{} & \cellQ{} & \multicolumn{2}{c|}{\cellSyntaxComp{}} &1 & 0 & any & \cellDec{}NP-complete & \cite{Miller00}\\

		\hline
		\cellR{} & \cellQ{} & \multicolumn{2}{c|}{any} &1 & 1 or 2 & any & \cellOpen{} & \cellOpenref{}\\

		\hline
		\cellR{} & \cellQboundedonetwo{} & \multicolumn{2}{c|}{\cellSyntaxComp{}} &1 & 3 & 1 & \cellUndec{}undecidable & \cite{Miller00}\\

		\hline
		\cellR{} & \cellQ{} & \multicolumn{2}{c|}{any} &2 & 0 or 1 & any & \cellOpen{} & \cellOpenref{}\\

		\hline
		\cellR{} & \cellQboundedonetwo{} & \multicolumn{2}{c|}{\cellSyntaxComp{}} &2 & 2 & 1 & \cellUndec{}undecidable & \cite{Miller00} (consequence)\\

		\hline
		\cellR{} & \cellQboundedonetwo{} & \multicolumn{2}{c|}{\cellSyntaxComp{}} &3 & 0 & 1 & \cellUndec{}undecidable & \cite{Miller00}\\
		\hline
		\cellR{} & \cellR{} & \cellSyntaxEq{} & \cellSyntaxNone{} & 3 & 0 & 6 & \cellUndec{}undecidable & \cite{AHV93}\\
		\hline
		\hline
		\cellR{} & \cellQ{} & \multicolumn{2}{c|}{\cellSyntaxOpen{}} & $< 2$ & $3$ & $2$ & \cellOpen{} & \cellOpenref{}\\
		\hline
		\cellR{} & \cellQ{} & \multicolumn{2}{c|}{\cellSyntaxOpen{}} & $2$ & $< 3$ & $2$ & \cellOpen{} & \cellOpenref{}\\
		\hline
		\cellR{} & \cellQ{} & \multicolumn{2}{c|}{\cellSyntaxOpen{}} & $2$ & $3$ & $< 2$ & \cellOpen{} & \cellOpenref{}\\
		\hline
		\cellQR{} & \cellQR{} & \multicolumn{2}{c|}{\cellSyntaxOpen{}} &2 & 3 & 2 & \cellUndec{}undecidable & \cite{Doyen07}\\
		\hline
	\end{tabular}
}

	\caption{Decidability of the EF-emptiness problem for general PTAs}
	\label{table:dec:EFemptiness}
\end{table*}

\ea{warning pour AHV93 nonelementary: en fait, leur construction n'utilise que des égalités… quid de BO14 ?}

We summarize these results in \cref{table:dec:EFemptiness}.\footnote{%
	This table is partially inspired by a similar table in \cite{Doyen07}, improved by adding more dimensions, and of course more recent results.
}
We only keep in \cref{table:dec:EFemptiness} the best known results as of the current state of the art.
For example, the decidability of the EF-emptiness problem over dense time with 1 parametric clock and arbitrarily many nonparametric clocks and integer-valued parameters as proved in~\cite{AHV93} with a non-elementary complexity does not appear in \cref{table:dec:EFemptiness} as it is subsumed by~\cite{BBLS15} with an NEXPTIME complexity and a more permissive syntax (use of invariants).

The open question of the syntax expressiveness requires to consider a multi-dimensional table: we need to consider not only the number of clocks and parameters, but also the syntax allowed in guards and invariants.
For example, %
for the undecidability over discrete time, \cite{BBLS15} improves the number of parameters when compared to~\cite{AHV93} (6 instead of~1), but requires both strict and non-strict inequalities, whereas \cite{AHV93} uses only equalities in their construction; it is therefore unclear whether the result of~\cite{AHV93} is really subsumed by~\cite{BBLS15}.
However, following \cref{remark:invariants}, we considered that the works requiring invariants to contain only clocks bounded from above impose in fact no constraint on the invariants form (as the clocks bounded from below can be moved to the incoming guards).

``Consequence'' indicates a result originally proved for a less expressive or a more expressive setting; ``at most'' in the complexity column indicates in the latter case that the complexity is necessarily lower or equal to that of the more expressive setting.
For example, \cite{Miller00} proved that the single clock case is PTIME over dense time with a fixed number of rational-valued parameters, and therefore the corresponding problem cannot be harder over discrete time (with integer-valued parameters).

In the following, we extract the most important results out of \cref{table:dec:EFemptiness}.
The decidability is clearly impacted by the number of parametric clocks, and we therefore reason by the number of parametric clocks.

\paragraph{Main results: 1 parametric clock}
	First, let us consider PTAs with a single parametric clock:
	The EF-emptiness problem is (at most) NP-complete over discrete and dense time with no nonparametric clock and arbitrarily many parameters~\cite{Miller00}.

	It is decidable and NEXPTIME-complete over discrete time with arbitrarily many nonparametric clocks~\cite{BBLS15}.
	\ea{Pourquoi NEXPTIME-complete ? Car \cite{BO14} montre NEXPTIME complete avec seulement inégalités larges ; or \cite{BBLS15} montre NEXPTIME avec une syntaxe potentiellement plus riche.}
	Over dense time with arbitrarily many nonparametric clocks and integer-valued parameters, it is NEXPTIME~\cite{BBLS15}.
	
	It is %
	undecidable with three nonparametric clocks~\cite{Miller00} over dense time with rational-valued parameters; note that this problem is decidable over discrete time~\cite{AHV93,BO14,BBLS15} and over dense time with integer-valued parameters~\cite{BBLS15}, which exhibits a difference between dense and discrete time~\cite{Miller00}, as well as between integer- and rational-valued parameters over dense time.

\paragraph{Main results: 2 parametric clocks}
	Second, let us consider PTAs with two parametric clocks:
	the EF-emptiness problem is decidable over discrete time with arbitrarily many nonparametric clocks and a single parameter, and is PSPACE$^\mathrm{NEXP}$-hard~\cite{BO14}\ea{pas sûr, en fait ``; this result is claimed in~\cite{BO14} to extend to dense time with integer-valued parameters''}.
	
	Over dense time with rational-valued parameters, the case with 2 parametric clocks and 2 nonparametric clocks is undecidable: \cite{Miller00} gives a proof of undecidability with 1 parametric clock and 3 nonparametric clocks: comparing one of the nonparametric clocks with a parameter in an additional location (\eg{} after the halting location) does not impact the proof and turns a nonparametric clock into a parametric one.
	
	Any other case with two parametric clocks remains open.

\paragraph{Main results: other undecidability}
The EF-emptiness problem is undecidable in all settings with three (or more) parametric clocks.

Finally, using only strict inequalities, the EF-emptiness problem is undecidable over dense time for two parametric clocks, three nonparametric clocks and two parameters~\cite{Doyen07}; this situation was not considered over discrete time.

\paragraph{Open cases}
The main open case is the ``two parametric clocks'' case.
The decidability is open for 2 parametric clocks with:
\begin{itemize}
	\item over discrete time: arbitrarily many nonparametric clocks and more than one parameter;
	\item over dense time with integer-valued parameters: arbitrarily many nonparametric clocks and parameters;
	\item over dense time with rational-valued parameters: 0 or 1 nonparametric clock and any number of parameters.
\end{itemize}

In addition, the decidability remains open over dense time with rational-valued parameters for 1 nonparametric clock, 1 or 2 nonparametric clocks and arbitrarily many parameters.

Finally, the decidability using only strict inequalities remain open for cases not considered by~\cite{Doyen07}: less clocks and parameters, or with integer-valued parameters (both over dense and discrete time).

\subsection{Language and trace preservation}

Let us first recall the definition of determinism from~\cite{AM15}.
We say that a PTA is \emph{deterministic} if, for any $\loc \in \Loc$, for any $\action \in \Actions$, there exists at most one edge $(\loc,\guard,\action,\resets,\loc') \in \Edges$, for some $\guard,\resets,\loc'$.
(Note that it differs from a rather common definition of determinism for TAs, that allows two or more outgoing transitions with the same action label provided that the corresponding guards are pairwise disjoint.)

The language- and trace-preservation problems are decidable for deterministic PTAs with a single (parametric) clock, and with linear parameter constraints allowed in guards and invariants, \ie{} of the form $\cellSyntaxJLRg$ or $\plterm \compOp 0$~\cite{AM15}.
A procedure to compute parameter valuations with the same trace set as a given valuation is proposed in~\cite{AM15} (close to the ``inverse method'' \cite{ACEF09}), that is complete for deterministic PTAs, and terminates in the case of a single clock.

\subsection{Parametric model checking}

Parametric model checking was addressed in different settings:
verifying a nonparametric model against a parametric formula,
or a parametric model against a nonparametric formula,
or a parametric model against a parametric formula.

\paragraph{Nonparametric model / parametric formula}
In~\cite{AETP01}, an extension of LTL with parameters in the formula (``PLTL'') is studied.
When only parametric ``always'' modalities are allowed of the form ``$\leq \param$'', checking emptiness of the valuation set is PSPACE-complete.
The solution to the synthesis problem is doubly exponential in the number of parameters.
However, when allowing equality in PLTL, the emptiness problem becomes undecidable~\cite{AETP01}.

\paragraph{Parametric model / nonparametric formula}
In~\cite{Quaas14}, it is shown that model checking PTAs with the (nonparametric) logic MTL~\cite{Koymans90} is undecidable, even with a single clock and a single parameter, and even when the PTAs is deterministic.
This negative result comes in contrast to the decidability of the EF-emptiness problem for one-clock PTAs, and to the decidability of MTL-model checking for (nonparametric) timed automata in the pointwise semantics over finite timed words~\cite{OW07}. %
Note that the proof of undecidability of~\cite{Quaas14} requires the parameters to be rational-valued (integer-valued parameters are not sufficient---and this latter case can hence be considered as open).

\paragraph{Parametric model / parametric formula}
Model checking a PTA over discrete time with a single parametric clock against a PTCTL formula (a parametric version of TCTL) is decidable, provided the formula does not use equality constraints; otherwise the problem becomes undecidable~\cite{BR07}.

\subsection{Other problems: open}

Other problems are open.
However, two constructions were recently proposed for the one parametric clock case, that may help solve most problems in this particular setting.
First, in~\cite{AM15}, we show that the parametric zone graph is finite for a single (parametric) clock and arbitrarily many rational-valued parameters over dense time.
This implies that all problems that reason on the zone graph can be decided.
This includes in particular EF-, EG-, AF and AG-emptiness, as well as the language and trace preservation problems.

Second, in~\cite{BBLS15}, an abstraction is proposed for one parametric clock and arbitrarily many nonparametric clocks and integer-valued parameters over dense time.
Although this remains to be shown formally, this abstraction (based on the elimination of the nonparametric clocks followed by a corner-point abstraction on the subsequent region graph) apparently preserves enough elements of the region graph to be used to solve all aforementioned problems.

In both cases, the synthesis seems also to be feasible.

\section[The (quite) disappointing class of L/U-PTAs]{The (quite) disappointing class of L/U-PTAs}\label{section:LU}

\reviewer{3}{Why are L/U PTAs disappointing? More problems are decidable for
L/U PTAs than for general PTAs. So you should rather call this section "The joy
of L/U PTAs"!}

Lower-bound/upper-bound parametric timed automata (L/U-PTAs), proposed in~\cite{HRSV02}, restrict the use of parameters in the model.
A parameter is said to be an \emph{upper-bound parameter} if, whenever it is compared with a clock, it is necessarily compared as an upper bound, \ie{} it only appears in inequalities of the form $\clock \compOpLeq \param$.
Conversely, a parameter is a \emph{lower-bound parameter} if it is only compared with clocks as a lower bound, \ie{} of the form $\param \compOpLeq \clock$.

An L/U-PTA is a PTA where the set of parameters is partitioned into upper-bound parameters and lower-bound parameters.
In~\cite{BlT09}, two additional subclasses are introduced: L-PTAs (resp.\ U-PTAs) are PTAs with only lower-bound (resp.\ upper-bound) parameters.

\begin{example}
	Consider again the coffee machine in \cref{fig:coffee}, modeled using a PTA~$\A$.
	This PTA is not an L/U-PTA; indeed, in the guard $\styleclock{x_2} = \styleparam{\param_2}$ (resp.\ $\styleclock{x_2} = \styleparam{\param_3}$), $\styleparam{\param_2}$ (resp.~$\styleparam{\param_3}$) is compared with clocks both as a lower-bound and as an upper-bound.
	(Recall that $=$ stands for $\leq$ and $\geq$.)
	
	However, if one replaces $\styleclock{x_2} = \styleparam{\param_2}$ with $\styleclock{x_2} \leq \styleparam{\param_2}$
	and one replaces $\styleclock{x_2} = \styleparam{\param_3}$ with $\styleclock{x_2} \leq \styleparam{\param_3}$, then $\A$ becomes an L/U-PTA with lower-bound parameter $\styleparam{\param_1}$ and upper-bound parameters $\{\styleparam{\param_2}, \styleparam{\param_3}\}$.
	Note that equalities are not forbidden in L/U-PTAs (\eg{} $\styleclock{x_1} = 10$), but only equalities involving parameters.
\end{example}

Several case studies fit into the class of L/U-PTAs: the root contention protocol, the bounded retransmission protocol and the Fischer mutual exclusion protocol are all modeled with L/U-PTAs in~\cite{HRSV02};
in \cite{HRSV02,KP12}, both the Fischer mutual exclusion protocol and a producer-consumer are verified using L/U-PTAs.
Interestingly, the two case studies of the seminal paper on PTAs~\cite{AHV93} (\viz{} a toy railroad crossing model and a model of Fischer mutual exclusion protocol) are also L/U-PTAs\LongVersion{, although the concept of L/U-PTAs had not been proposed yet at that time}.
In addition, most models of asynchronous circuits with bi-bounded delays (\ie{} where each delay between the change of an input signal and the change of the corresponding output is a parametric interval) can be modeled using L/U-PTAs.

L/U-PTAs were first known for their decidability results (\cref{ss:L/U:decidability});
then, new undecidability results (\cref{ss:L/U:undecidability,ss:L/U:open}) rendered this class less interesting.
The most disappointing aspect of L/U-PTAs is the impossibility to perform exact synthesis even when the associated decision problems are decidable.
We review these results in the remainder of this section.

\subsection{A main decidability result}\label{ss:L/U:decidability}

The first (and main) positive result for L/U-PTAs is the decidability of the EF-emptiness problem~\cite{HRSV02}.
L/U-PTAs benefit from the following interesting monotonicity property: increasing the value of an upper-bound parameter or decreasing the value of a lower-bound parameter necessarily relaxes the guards and invariants, and hence can only add behaviors.
Hence, checking the EF-emptiness of an L/U-PTA can be achieved by replacing all lower-bound parameters with~0, and all upper-bound parameters with $\infty$; this yields a nonparametric TA, for which emptiness is PSPACE~\cite{AD94}.
This procedure is not only sound but also complete.

\ea{TODO: ajouter le truc entièrement paramétré}

Further decidability results are exhibited in~\cite{BlT09}, for infinite runs acceptance properties, \ie{} where a location is met infinitely often (a problem to which we refer hereafter as \BuchiEF{}).
Note that, in contrast to~\cite{HRSV02} where the parameters are valued with non-negative reals, the results in~\cite{BlT09} consider integer-valued parameters (though time is dense, \ie{} clocks are real-valued).
It is shown in~\cite{BlT09} that \BuchiEF{}-emptiness, \BuchiEF{}-universality, and \BuchiEF{}-finiteness are PSPACE-complete.
Remark that the decidability of the \BuchiEF{}-finiteness is due to the fact that the parameters are integer-valued; in short, a sufficient bound is computed on the parameters, and then valuations smaller or equal to this bound are enumerated, which would not be feasible for real- or rational-valued parameters.\ea{note: only for L-PTAs; if a general L/U-PTA, emptiness is equivalent to finiteness}

Oddly, the decidability of EF-universality was never shown for L/U-PTAs.
On the one hand, EF-emptiness is decidable for L/U-PTAs with rational-valued parameters~\cite{HRSV02}.
On the other hand, \BuchiEF{}-universality is decidable for L/U-PTAs with integer-valued parameters~\cite{BlT09}, and this result extends in a very straightforward manner to EF-universality for L/U-PTAs with integer-valued parameters.
Let us first extend \BuchiEF{}-universality to rational-valued parameters.
The result mainly consists in a reasoning dual to~\cite[Lemma~2]{ALime17}.

\begin{proposition}\label{proposition:BuchiEF-universality}%
	The \BuchiEF{}-universality problem is PSPACE-complete for L/U-PTAs with rational-valued parameters.
\end{proposition}
\begin{proof}
	We aim at proving that, given an L/U-PTA $\A$ and a subset of its locations~$\somelocs$, the problem of the universality of the set of parameter valuations~$\pval$ such that $\pval(\A)$ has a run passing infinitely often through~$\somelocs$ is PSPACE-complete.
	Let us prove that the set of rational valuations satisfying the property is not universal iff the set of integer valuations doing so is not universal.
	
	\begin{itemize}
		\item[$\Leftarrow$] Considering that integer valuations are also rational valuations, the result trivially holds.
		\item[$\Rightarrow$]
			Assume there exists a rational-valued parameter valuation~$\pval$ for which $\valuate{\A}{\pval}$ contains \emph{no} infinite run passing infinitely often through locations of~$\somelocs$.
		Let $\pval'$ be the integer parameter valuation obtained from~$\pval$ as follows:
		\[\pval'(\param) = 
			\begin{cases}
			\pval(\param) & \text{if }\pval(\param) \in \grandn \\
			\floor{\pval(\param)} & \text{if }\param\text{ is an upper-bound parameter}\\
			\ceil{\pval(\param)} & \text{if }\param\text{ is a lower-bound parameter}\\
			\end{cases}
		\]
		
		That is, $\pval'$ is more restrictive than~$\param$, and less guards will be enabled, and therefore less behaviors will be possible in~$\valuate{\A}{\pval}$.
		Formally, from the well-known monotonicity property of L/U-PTAs (recalled in \eg{} \cite[Lemma~1]{ALime17}), if $\valuate{\A}{\pval}$ yields no infinite run passing infinitely often through locations of~$\somelocs$, then neither does $\valuate{\A}{\pval'}$.
	\end{itemize}
	Now, in~\cite[Theorem~8]{BlT09}, it is proved that the problem of the universality of the set of integer parameter valuations for which there exists an infinite run passing infinitely often through~$\somelocs$ is PSPACE-complete.
	This concludes the proof.
\end{proof}

This result extends trivially to EF-universality (by adding self-loops with no guard on all accepting locations).

\begin{corollary}\label{corollary:EF-universality} %
	The EF-universality problem is PSPACE-complete for L/U-PTAs with rational-valued parameters.
\end{corollary}
\subsection{Undecidability results}\label{ss:L/U:undecidability}

The first undecidability results for L/U-PTAs are shown in~\cite{BlT09}:
the \emph{constrained} \BuchiEF{}-emptiness problem and constrained \BuchiEF{}-universality problem are undecidable for L/U-PTAs.
By constrained it is meant that some parameters of the L/U-PTA can be constrained by an initial linear constraint, \eg{} $\param_1 \leq 2 \times \param_2 + \param_3$.
Indeed, using linear constraints, one can constrain an upper-bound parameter to be equal to a lower-bound parameter, and hence build a 2-counter machine using an L/U-PTA.
However, when no upper-bound parameter is compared to a lower-bound parameter (\ie{} when no initial linear inequality contains both an upper-bound and a lower-bound parameter), these two problems retrieve decidability~\cite{BlT09}.
The exact decidability frontier may not be found yet: the case where a lower-bound parameter is constrained to be less than or equal to an upper-bound parameter fits in none of the considered cases.

A second negative result is shown in~\cite{JLR15}: the AF-emptiness problem is undecidable for L/U-PTAs.
This is achieved by a reduction from a 2-counter machine where a lower-bound parameter is equal to an upper-bound parameter iff AF holds.
This restricts again the use of L/U-PTAs, as AF is essential to show that all possible runs of a system eventually reach a (good) state.

Then, in \cite{AM15}, it is shown that the language preservation problem is undecidable for L/U-PTAs.
Again, this is achieved by a reduction from a 2-counter machine where a lower-bound parameter is equal to an upper-bound parameter iff the language is preserved.

\subsection{A frontier between decidability and undecidability}\label{ss:L/U:open}

The EG-emptiness problem stands at the frontier between decidability and undecidability~\cite{ALime17}.
Recall that the EG-emptiness problem is false if there exists at least one parameter valuation for which a maximal run remains entirely within some predefined set~$\somelocs$ of locations.
That is, either this run is an infinite run, and therefore contains a cycle (remaining within~$\somelocs$);
or this run is a finite run (remaining within~$\somelocs$), and therefore ends with a deadlock, \ie{} ends with a state from which no discrete transition can be taken, even after letting some time elapse.

On the one hand, deciding whether there exists a valuation in an L/U-PTA yielding a cycle is decidable (and PSPACE-complete).
On the other hand, deciding whether there exists a valuation in an L/U-PTA yielding a deadlock is undecidable.
(These two problems, not studied in this survey, are without surprise shown to be undecidable for general PTAs.)

The EG-emptiness problem stands in between decidability and undecidability:
while this problem is decidable for L/U-PTAs with a bounded parameter domain with closed bounds, it becomes undecidable if either the assumption of boundedness or of closed bounds is lifted.

\subsection{Model-checking L/U-PTAs}

In~\cite{BlT09}, a parametric extension of the dense-time linear temporal logic \mitlzeroinf{} (denoted ``\pmitlzeroinf{}'') is proposed; when parameters are used only as lower or upper bound in the formula (to which we refer as L/U-\pmitlzeroinf{}), satisfiability and model checking are PSPACE-complete; this is obtained by translating the formula into an L/U-PTA and checking an infinite acceptance property.

Then, in~\cite{GLN15}, an extension of MITL allowing parametric linear expressions in bounds is proposed (yielding PMITL).
Two sets of (integer-valued) parameter valuations are considered: 
\begin{inparaenum}[\itshape 1\upshape)]
	\item the set of valuations for which a PMITL formula is satisfiable, \ie{} for which there exists a timed sequence (possibly belonging to a given L/U-PTA) satisfying it, and
	\item the set of valuations for which a PMITL formula is valid, \ie{} for which all timed sequences (possibly belonging to a given L/U-PTA) satisfy it.
\end{inparaenum}
Under some assumptions, the emptiness and universality of the valuation set for which a PMITL property is satisfiable or valid (possibly \wrt{} a given L/U-PTA) are decidable, and EXPSPACE-complete.
Essential assumptions for decidability include the fact that parameters should be used with the same polarity (positive or negative coefficient, as lower or upper bound in the intervals) within the entire PMITL formula, and each interval can only use parameters in one of the endpoints.
Additional assumptions include that no interval of the PMITL formula should be punctual (nor empty), and linear parametric expressions are only used in right endpoints of the intervals (single parameters can still be used as left endpoints).
In addition, two fragments of PMITL are showed to be in PSPACE, including one that allows for expressing parameterized response (``if an event occurs, then another event shall occur within some possibly parametric time interval'').

\subsection{Summary of decidability problems for L/U-PTAs}
We summarize in \cref{table:L/U} decision problems for L/U-PTAs.
Cases not considered in the literature %
are not depicted.

\begin{table*}%

	\centering
	\begin{tabular}{| c | c | c | c |}
		\hline
		\cellHeader{Problem} & \cellHeader{\DomainP{}} & \cellHeader{Complexity} & \cellHeader{Main ref.}\\
		\hline
		EF-emptiness & \cellQ{} & \cellDec{}PSPACE-complete & \cite{HRSV02}\\
		\hline
		AG-emptiness & \cellQ{} & \cellDec{}PSPACE-complete & \cref{corollary:EF-universality}\\
		\hline
		AF-emptiness & \cellQ{} & \cellUndec{}undecidable & \cite{JLR15}\\
		\hline
		\hline
		Cycle-existence-emptiness & \cellQ{} & \cellDec{}decidable & \cite{ALime17}\\
		\hline
		Deadlock-existence-emptiness & \cellQ{} & \cellUndec{}undecidable & \cite{ALime17}\\
		\hline
		EG-emptiness (closed bounded) & \cellQ{} & \cellDec{}decidable & \cite{ALime17}\\
		\hline
		EG-emptiness (general) & \cellQ{} & \cellUndec{}undecidable & \cite{ALime17}\\
		\hline
		\hline
		\BuchiEF{}-emptiness & \cellQ{} & \cellDec{}PSPACE-complete & \cite{ALime17}\\
		\hline
		\BuchiEF{}-universality & \cellQ{} & \cellDec{}PSPACE-complete & \cref{proposition:BuchiEF-universality}\\
		\hline
		\BuchiEF{}-finiteness & \cellN{} & \cellDec{}PSPACE-complete & \cite{BlT09}\\
		\hline
		constrained \BuchiEF{}-emptiness & \cellN{} & \cellUndec{}undecidable & \cite{BlT09}\\
		\hline
		constrained \BuchiEF{}-universality & \cellN{} & \cellUndec{}undecidable & \cite{BlT09}\\
		\hline
		L/U-constrained \BuchiEF{}-emptiness & \cellN{} & \cellDec{}PSPACE-complete & \cite{BlT09}\\
		\hline
		L/U-constrained \BuchiEF{}-universality & \cellN{} & \cellDec{}PSPACE-complete & \cite{BlT09}\\
		\hline
		\hline
		Language preservation & \cellN{} & \cellUndec{}undecidable & \cite{AM15}\\
		\hline
		Language preservation & \cellQ{} & \cellUndec{}undecidable & \cite{AM15}\\
		\hline
		\hline
		L/U-\pmitlzeroinf{}-emptiness & \cellN{} & \cellDec{}PSPACE-complete & \cite{BlT09}\\
		\hline
		L/U-\pmitlzeroinf{}-universality & \cellN{} & \cellDec{}PSPACE-complete & \cite{BlT09}\\
		\hline
		PMITL model-checking & \cellN{} & \cellDec{}EXPSPACE-complete & \cite{GLN15}\\
		\hline
	\end{tabular}

	\caption{Decision problems for L/U-PTAs over dense time}
	\label{table:L/U}
\end{table*}
\subsection{Intractability of the synthesis}

The most disappointing result concerning L/U-PTAs is shown in~\cite{JLR15}:
despite decidability of the underlying decision problem (EF-emptiness),
the solution to the EF-synthesis problem for L/U-PTAs cannot be represented using a formalism for which the emptiness of the intersection with equality constraints is decidable.
The proof relies on the undecidability of the constrained emptiness problem of~\cite{BlT09}.
A very annoying consequence is that such a solution cannot be represented as a finite union of polyhedra (since the emptiness of the intersection with equality constraints is decidable).

\subsection{Two open classes: L-PTAs and U-PTAs}

L-PTAs and U-PTAs (introduced in~\cite{BlT09}) are very open classes, in the sense that %
to the best of our knowledge,
	no result known to be decidable for L-PTAs (or U-PTAs) was shown undecidable for L/U-PTAs (and is hence either decidable or open).
Conversely, and even stronger, no result known to be undecidable for L/U-PTAs was shown decidable for L-PTAs (or U-PTAs)---and remains open.

To summarize, the EG-emptiness, AG-emptiness and AF-emptiness problems, as well as the language- and trace-preservation problems, are all undecidable for (general) L/U-PTAs, but remain open for L-PTAs and U-PTAs.

\ifdefined\LaTeXdiff
\else
\ea{
\subsubsection{Decision problems}

\ea{Dans \cite{BlT09}, les résultats suivants s'étendent trivialement à R :
\begin{itemize}
	\item Proposition 3 For lower bound automata A, the set $\Gamma(A)$ is downward-closed.
	\item Theorem2 ForlowerboundautomataA,checkingemptinessof $\Gamma(A)$ isPSPACE-complete and can be done in time
	$O(|\Delta| \cdot (2 c_\A +2)^{2|\Clock|^2} )$.
\end{itemize}

}

\paragraph{\BuchiEF{}-emptiness for L-PTAs}
This is decided by checking whether the TA valuated with 0 is empty or not.
This can be lifted immediately to~\grandrplus{}.
As for \grandrplus{}, it can be decided by checking whether the TA valuated with $-\infty$ is empty or not, which leaves the complexity unchanged.
\ea{Note: we should still prove that the fact that $\Gamma(\A)$ is downward-closed in L-PTAs can be lifted to $\grandrplus{}$, but this looks obvious. (Yes, it is.)}

\paragraph{\BuchiEF{}-universality for L-PTAs}

\paragraph{\BuchiEF{}-emptiness for U-PTAs}

\paragraph{\BuchiEF{}-universality for U-PTAs}
This is decided by checking whether the TA valuated with 0 is empty or not.
This can be lifted immediately to~\grandrplus{}.
As for \grandr{}, it can be decided by checking whether the TA valuated with $-\infty$ is empty or not, which leaves the complexity unchanged.
\ea{Note: we should prove this time that the fact that $\Gamma(\A)$ is upward-closed in U-PTAs can be lifted to $\grandrplus{}$, but this looks obvious. (Yes, it is.)}

\begin{table*}
	\centering
	\begin{tabular}{| c | c | c | c | c |}
		\hline
		\cellHeader{Problem} & \cellHeader{\DomainP{}} & \cellHeader{L/U-PTAs} & \cellHeader{L-PTAs \& U-PTAs} & \cellHeader{Main ref.}\\
		\hline
		\BuchiEF{}-emptiness & \cellN{} &  \cellDec{}PSPACE-complete & \cellDec{}PSPACE-complete & \cite{BlT09}\\
		\hline
		\BuchiEF{}-universality & \cellN{} &  \cellDec{}PSPACE-complete & \cellDec{}PSPACE-complete & \cite{BlT09}\\
		\hline
		\BuchiEF{}-finiteness & \cellN{} &  \cellDec{}PSPACE-complete & \cellDec{}PSPACE-complete & \cite{BlT09}\\
		\hline
		Constrained-\BuchiEF{}-finiteness & \cellN{} &  \cellUndec{}undecidable & \cellDec{}PSPACE-complete & \cite{BlT09}\\
		\hline
		Constrained-\BuchiEF{}-universality & \cellN{} &  \cellUndec{}undecidable & \cellDec{}PSPACE-complete & \cite{BlT09}\\
		\hline
		AF-emptiness & \cellN{} & \cellUndec{}undecidable & \cellOpen{} & \cellOpenref{} \\
		\hline
		Any except EF-emptiness & \cellR{} & (depends) & \cellOpen{} & \cellOpenref{} \\
		\hline
	\end{tabular}

	\caption{Decision problems for L-PTAs and U-PTAs}
\end{table*}
\begin{table*}
	\centering
	\begin{tabular}{| c | c | c | c |}
		\hline
		\cellHeader{Problem} & \cellHeader{Over \grandn{}} & \cellHeader{Over \grandrplus{}} & \cellHeader{Over \grandr{}} \\
		\hline
		\BuchiEF{}-emptiness for L-PTAs & \cellDec{}PSPACE-complete & \cellDec{}PSPACE-complete & \cellDec{}PSPACE-complete\\
		\hline
		\BuchiEF{}-universality for L-PTAs & \cellDec{}PSPACE-complete & ? & ? \\
		\hline
		\BuchiEF{}-emptiness for U-PTAs & \cellDec{}PSPACE-complete & ? & ? \\
		\hline
		\BuchiEF{}-universality for U-PTAs & \cellDec{}PSPACE-complete & \cellDec{}PSPACE-complete & \cellDec{}PSPACE-complete\\
		\hline
	\end{tabular}

	\caption{Extension of \cite{BlT09} to \grandrplus{} and \grandr{}}
\end{table*}

}

\fi
\paragraph{Synthesis}
The synthesis for L-PTAs and U-PTAs did not receive much attention, with the exception of integer-valued parameters: in that case, it is possible to synthesize the solution to the \BuchiEF{}-synthesis problem in the form of a union of linear constraints doubly exponential in the number of parameters~\cite{BlT09}.
The authors note that it remains open whether one can construct a linear constraint with a single exponential blow-up.
This result does not extend in a straightforward manner to rational-valued parameters, as the technique in~\cite{BlT09} (for U-PTAs) requires the computation of a sufficient upper bound, and then an exhaustive enumeration of parameters below this bound.

\section[Beyond parametric timed automata]{Beyond parametric timed automata}\label{section:beyond}
\subsection{Parametric hybrid automata}

Hybrid automata~\cite{ACHH92,ACHHNHOSY95,Henzinger96} are an extension of timed automata where clocks (called continuous variables) can have an arbitrary rate (\ie{} non-necessarily equal to~1).

The reachability of a location in linear hybrid automata is undecidable, although semi-algorithms were proposed~\cite{ACHHNHOSY95}.
Interestingly, the simple extension of timed automata to stopwatch automata (where the elapsing of some clocks may be stopped in selected locations) yields a formalism as expressive as linear hybrid automata~\cite{CL00}, and for which reachability is undecidable too.

First, remark that parameters can be encoded naturally in the general class of hybrid automata, provided diagonal constraints are allowed (of the form $v_i - v_j \compOp c$, with $v_i,v_j$ variables and $c$ a constant): a parameter is a variable that is not initialized (its initial value is arbitrary), the rate of which is always~0 (therefore constant), and that is never reset along any transition.
However, the undecidability results for linear hybrid automata rule out the possibility of exhibiting any decidability results for (general) parametric linear hybrid automata.

Second, several subclasses of linear hybrid automata were defined in the literature, and were shown to enjoy some decidable results (\eg{} \cite{HKPV98,BMRT04,AMPS12,BHS12,BDGORW13}).
However, obviously, any such subclass at least expressive as timed automata (such as \cite{HKPV98,BDGORW13}) would necessarily lead to undecidability when adding parameters.
This is not the case of some of subclasses of linear hybrid automata, which are incomparable (at least from a syntactic point of view) with timed automata (\eg{} \cite{BMRT04,BHS12}), or restrict the use and the number of variables~\cite{AMPS12}.
We believe studying parametric extensions of these formalisms represent an interesting direction of research.

\subsection{Parametric interrupt timed automata}

Interrupt timed automata (ITAs) are a subclass of hybrid systems where clock variables only have a rate of~0 (stopped) or~1 (processing): in fact, ITAs define levels such that, at each level, exactly one clock is active (rate~1), while clocks of lower levels are stopped (rate~0)~\cite{BHS12}.
In addition, guards can only involve clocks from the current level and the lower levels.
Clock updates allow the use of linear expressions involving clocks from lower levels.
The model is well-suited to define real-time systems with multiple tasks running on a single processor and subject to interruption (where a lower-priority task can be preempted by a higher-priority task).
A main positive result for ITAs is that reachability is in NEXPTIME (and in PTIME when the number of clocks is fixed).
Interrupt timed automata and timed automata are incomparable in terms of timed language.

In~\cite{BHJL16}, ITAs are extended with parameters, which yields parametric ITAs (PITAs).
When parameters are combined with clock values in linear expressions as additive coefficients, the reachability in PITAs reduces to the same problem in nonparametric ITAs, and is therefore decidable (with an upper bound of 2EXPTIME on the complexity, due to the reduction).
When parameters are combined with clock values in linear expressions as both additive and multiplicative coefficients, the reachability in PITAs remains decidable, with an upper bound of 2EXPSPACE on the complexity.
This significantly increases the expressiveness of ITAs, and allows to model clock drifts.

Finally, ITAs are extended to polynomial ITAs (PolITAs) in~\cite{BHPSS15}, where polynomial expressions on clocks are allowed in guards.
Reachability remains decidable, and parameters can be used (without harming the complexity) in polynomials.

\subsection{Integer-point parametric timed automata}\label{ss:IPPTA}

Integer-point parametric timed automata (IP-PTAs) were introduced in~\cite{ALR16ICFEM} as a subclass of PTAs in which each state in the parametric zone graph (a construction with location and symbolic convex constraints over $\Clock \cup \Param$) contains an integer point.
The main positive result for IP-PTAs with bounded (rational-valued) parameters is the decidability of the EF-emptiness problem.
However, the AF-emptiness and AG-emptiness problems are both undecidable.

A more disappointing result is the undecidability of the membership problem, \ie{} it is undecidable whether a PTA is an IP-PTA.
In addition, synthesis is proved to be intractable (as for L/U-PTAs).

However, a sufficient syntactic condition for the membership of IP-PTAs is the class of \emph{reset-PTAs}~\cite{ALR16ICFEM}:
	whenever a clock is compared to a parameter in a transition guard (resp.\ in location invariant), then all clocks must be reset on that transition (resp.\ along all transitions going out from that location).
As a consequence, the EF-emptiness problem is decidable for bounded reset-PTAs too.
In addition, we conjecture that the parametric zone graph of reset-PTAs should be finite, which would allow to prove the decidability of the AF, AG and EG-emptiness problems.
This remains to be shown formally.

\subsection{Other formalisms}
\paragraph{Time Petri nets}
In parallel to timed automata, many works in the literature were dedicated to time Petri nets~\cite{Merlin74}, which are an extension of Petri nets where transitions are labeled with a firing interval, which represents the duration between the time when the transition becomes enabled (enough tokens are present in the incoming places) and the time it can actually fire.
Time Petri nets and timed automata were compared in, \eg{} \cite{BCHLR05,Srba08,BCHLR13}.

In~\cite{TLR09}, time Petri nets are extended with rational-valued parameters in firing intervals.
Using translations between time Petri nets and timed automata~\cite{BCHLR05,CR06}, it is shown that the emptiness and reachability problems are undecidable for bounded parametric time Petri nets, and turn decidable when parameters are only used as lower-bounds or upper-bounds, in the spirit of L/U-PTAs.
Then, semi-algorithms are defined for the parametric model-checking of a subset of parametric TCTL formulas applied on parametric time Petri nets extended with inhibitor arcs (which play a similar role as stopwatches in timed automata).
The tool \romeo{} implements these algorithms.

\paragraph{Stateful timed CSP}
In~\cite{ALSD12}, the process algebra stateful timed CSP~\cite{SLDLSA13} (itself an extension of Hoare's communicating sequential processes~\cite{Hoare78}) is extended with parameters in syntactic constructs such as \texttt{Wait}, \texttt{Deadline} or \texttt{Within}, yielding PSTCSP.
Without surprise (as the expressiveness of PSTCSP is very close to that of PTAs), the emptiness of the valuation set for which a configuration is reachable is undecidable.
Although most of the (timed) syntactic constructs allowed are not necessary for the proof of undecidability, the \texttt{Wait} construct (used to test an exact amount of time, similar to equality in timed automata) is extensively used.
Decidability for subsets of the syntax without the \texttt{Wait} construct was not studied.
PSTCSP is implemented in a tool \PSyHCoS{}~\cite{ALSDL13} implementing some parameter synthesis algorithms.

\section[Tools and applications]{Tools and applications}\label{section:tools}
\subsection{Tools}

The first tool to support modeling and verification using parametric timed automata was \hytech{}~\cite{HHW97}.
In fact, \hytech{} supports linear hybrid automata (including clocks, parameters, stopwatches and general continuous variables); it can compute the state space, and perform operations (such as intersection, convex hull, difference) between sets of symbolic states.
Therefore, it can be used to perform parametric model checking using reachability checking~\cite{ABBL03}.
\hytech{} is not maintained anymore, but can still be found online in the form of a standalone binary for Linux.\footnote{%
	\url{https://embedded.eecs.berkeley.edu/research/hytech/}
}

In~\cite{HRSV02}, an extension of \uppaal{} implementing parametric difference bound matrices (PDBMs) and hence allowing for verification using PTAs is mentioned.
However, this tool does not seem to be available anywhere online.

\romeo{}~\cite{LRST09} primarily supports parametric time Petri nets (extended with stopwatches), a formalism shown to be close to PTAs in terms of expressiveness~\cite{BCHLR05,TLR09}.
\romeo{} supports the use of parametric linear expressions in the time intervals of the transitions, and allows to add linear constraints on the parameters to restrict their domain.
\romeo{} also implements an original algorithm for integer parameter synthesis using a symbolic (continuous) representation~\cite{JLR15}.
In addition, \romeo{} provides a simulator and an integrated model-checker supporting a subset of the TCTL syntax (including EF-synthesis and AF-synthesis).
\romeo{} is mainly written in C++, and makes use of the Parma Polyhedra Library~\cite{BHZ08}.

\imitator{}~\cite{AFKS12} is a software tool for parametric verification and robustness analysis of PTAs augmented with integer variables and stopwatches.
Parameters can be used both in the model and in the properties.
Verification capabilities include EF-synthesis, deadlock-freeness-synthesis~\cite{Andre16}, non-Zeno model checking~\cite{ANPS17}, and trace-preservation-synthesis.
\imitator{} is fully written in OCaml, and makes use of the Parma Polyhedra Library~\cite{BHZ08}.
It also features distributed capabilities to run over a cluster.

\subsection{Applications}

The formalism of PTAs has been used to model and verify various case studies featuring real-time constraints and parameters.

Beyond the usual academic examples (such as variants of train controllers~\cite{AHV93,HRSV02}),
PTAs were also used to successfully specify and verify numerous interesting case studies such as
the root contention protocol~\cite{HRSV02},
Philip's bounded retransmission protocol~\cite{HRSV02},
a 4-phase handshake protocol~\cite{KP12},
the alternating bit protocol~\cite{JLR15},
an asynchronous circuit commercialized by ST-Microelectronics~\cite{CEFX09},
(non-preemptive) schedulability problems~\cite{JLR15},
a distributed prospective architecture for the flight control system of the next generation of spacecrafts designed at ASTRIUM Space Transportation~\cite{FLMS12},
an unmanned aerial video system by Thales\footnote{%
	\url{https://www.imitator.fr/static/FMTV15/}
}, %
and even analysis of music scores~\cite{FJ13}.
\section[Open questions and perspectives]{Open questions and perspectives}\label{section:open}

\paragraph{Syntax and expressiveness}
A first perspective is to compare the expressiveness of the various syntaxes of guards and invariants for general PTAs used in the literature.
The definitions of expressiveness recently proposed in~\cite{ALR16ICFEM} could be reused for that purpose, using either untimed or timed languages.
Comparing the expressiveness of the syntaxes in the literature would reduce the number of dimensions for the various decidability results of the EF-emptiness problem studied in \cref{table:dec:EFemptiness}.

\paragraph{Open decision problems}
A main open problem is the decidability of PTAs with two parametric clocks, that was only studied with a single integer-parameter~\cite{BO14}.
Studying further the EG-, AF- and AG-emptiness problems for few clocks and parameters (as it was quite extensively done for EF-emptiness) remains to be done too, although the practical interest may be somehow debatable.

In addition, with the exception of~\cite{AM15}, all proofs of undecidability in the literature use a bounded number of clocks and parameters, but an unbounded number of locations.
Exhibiting a minimal number of locations (at the possible cost of an unbounded number of variables) may be of theoretical interest.
\ifdefined\LaTeXdiff
\else
\ea{
what with negative parameters?

major open problems:
\begin{enumerate}
	\item open PTAs with integer parameters; est-ce que je ne me ferais pas ça tout seul ? ;-) (mais pas évident du tout pour les horloges… ; en fait facile pour N/N, mais pas pour Q/N)
\end{enumerate}

}
\fi

More interesting (and promising) are the two open classes of L-PTAs and U-PTAs.
These classes are non-trivial, and relate to the robust analysis of TAs: most robustness problems (see~\cite{Markey11,BMS13}) consider an enlargement of all guards by (usually) the same constant factor, whereas U-PTAs allow to enlarge or decrease \emph{some} of the upper-bound guards by a possibly different rational-valued parameter, which gives an orthogonal definition of robustness.
The language preservation problem remains open for U-PTAs~\cite{AM15} (except in the case of a single integer-valued parameter where it becomes decidable), and the question of the synthesis is also challenging.

\paragraph{Hidden decidable subclasses?}
Despite many undecidability problems, PTAs were often used to model and verify various case studies (see \cref{section:tools}).
This can be seen as a paradox considering the numerous undecidability results PTAs suffer from.
In fact, as the aforementioned analyses terminate almost always with an exact result, it is challenging to understand why, and perhaps to exhibit further classes for which the problems considered in this survey become decidable.

\paragraph{Hybrid systems with parameters}

Some subclasses of linear hybrid automata are incomparable with timed automata (\eg{} \cite{BMRT04,AMPS12}), and parametric extensions could be studied.
Recall that the class of interrupt timed automata benefits from decidability results even when extended with parameters~\cite{BHPSS15}.

\paragraph{Synthesis}
Whereas decision problems (considered in this document) were much studied, little interest has been dedicated to the synthesis of parameters, which should, however, be a main practical challenge.
Despite undecidability (in general~\cite{AHV93}) or intractability (for L/U-PTAs~\cite{JLR15}), semi-algorithms or approximated procedures could be devised; SMT-based techniques~\cite{KP12}, or the integer hull approximation~\cite{JLR15,ALR15} can serve as a basis for future works.
Also note that two recent orthogonal works aimed at performing synthesis in a compositional manner~\cite{ABBCR16,ALin17}.

Also, combining nonparametric analysis (\eg{} with the efficient model checker \uppaal{}) with parametric analysis, so as to find perhaps not all valuations, but at least some of them, is certainly a promising direction of research.

\ea{different kinds of parameters: see ``Parametric analysis of computer systems'' (WH97) and the case study with Nx… for a protocol (BRP??)}

\section*{Acknowledgements}
This manuscript benefited from discussions with Didier Lime, Nicolas Markey, and Olivier H. Roux,
as well as from the useful comments and suggestions of all three anonymous reviewers.

	\newcommand{\CCIS}{Communications in Computer and Information Science}
	\newcommand{\ENTCS}{Electronic Notes in Theoretical Computer Science}
	\newcommand{\FAC}{Formal Aspects of Computing}
	\newcommand{\FI}{Fundamenta Informaticae}
	\newcommand{\FMSD}{Formal Methods in System Design}
	\newcommand{\IJFCS}{International Journal of Foundations of Computer Science}
	\newcommand{\IJSSE}{International Journal of Secure Software Engineering}
	\newcommand{\IPL}{Information Processing Letters}
	\newcommand{\JLAP}{Journal of Logic and Algebraic Programming}
	\newcommand{\JLAMP}{Journal of Logical and Algebraic Methods in Programming} %
	\newcommand{\JLC}{Journal of Logic and Computation}
	\newcommand{\LMCS}{Logical Methods in Computer Science}
	\newcommand{\LNCS}{Lecture Notes in Computer Science}
	\newcommand{\RESS}{Reliability Engineering \& System Safety}
	\newcommand{\STTT}{International Journal on Software Tools for Technology Transfer}
	\newcommand{\TCS}{Theoretical Computer Science}
	\newcommand{\ToPNoC}{Transactions on Petri Nets and Other Models of Concurrency}
	\newcommand{\TSE}{{IEEE} Transactions on Software Engineering}

	\renewcommand*{\bibfont}{\small}
	\printbibliography[title={References}]
\ifdefined\WithReply
	\clearpage
	\input{letter1.tex}
\fi

\end{document}